\documentclass[11pt]{article} 

\usepackage{fullpage} 
\usepackage{graphicx}
\usepackage{subcaption}
\usepackage{upref}
\usepackage{enumerate}
\usepackage{latexsym}
\usepackage{color,graphics}
\usepackage{relsize}
\usepackage{float}
\restylefloat{table}
\usepackage{algorithm}
\usepackage[redeflists]{IEEEtrantools}
\usepackage{centernot}
\usepackage{mathtools}
\usepackage{stmaryrd}
\usepackage{amsmath,amsthm,amssymb,algorithmic,epsfig,graphicx, tikz}
\usepackage{amsfonts}
\usepackage{varioref}
\usepackage[ansinew]{inputenc}
\usepackage{qcircuit}
\usepackage{amsmath}
\usepackage{amsfonts}
\usepackage{tikz} 
\usepackage{multirow}
\usepackage{bbm}
\usepackage{amssymb}
\usepackage{varioref}
\usepackage[ansinew]{inputenc}
\usepackage{authblk}
\usepackage{multicol}
\usepackage{textcomp}
\usepackage{booktabs}
\usepackage{multirow}
\newcommand{\norm}[1]{\left\lVert#1\right\rVert}
\newtheorem{theorem}{Theorem}[section]
\newtheorem{lemma}[theorem]{Lemma}

\newtheorem{result}{Result}
\newtheorem{assumption*}{Assumption}
\newtheorem{definition}{Definition}
\newcommand{\R}{\mathbb{R}}

\newcommand{\E}{\mathbb{E}}
\newcommand{\PP}{\mathbb{P}}

\def\ket#1{\mathinner{|{#1}\rangle}}
\newcommand{\braket}[2]{\langle #1|#2\rangle}
\renewcommand{\part}[2]{\frac{\partial #1}{\partial #2}}

\DeclarePairedDelimiter\ceil{\lceil}{\rceil}
\DeclarePairedDelimiter\floor{\lfloor}{\rfloor}

\makeatletter
\newcommand{\xMapsto}[2][]{\ext@arrow 0599{\Mapstofill@}{#1}{#2}}
\def\Mapstofill@{\arrowfill@{\Mapstochar\Relbar}\Relbar\Rightarrow}
\makeatother

\newcommand{\thmref}[1]{\hyperref[#1]{{Theorem~\ref*{#1}}}}
\newcommand{\lemref}[1]{\hyperref[#1]{{Lemma~\ref*{#1}}}}
\newcommand{\remref}[1]{\hyperref[#1]{{Remark~\ref*{#1}}}}
\newcommand{\corref}[1]{\hyperref[#1]{{Corollary~\ref*{#1}}}}
\newcommand{\eqnref}[1]{\hyperref[#1]{{Equation~(\ref*{#1})}}}
\newcommand{\claimref}[1]{\hyperref[#1]{{Claim~\ref*{#1}}}}
\newcommand{\remarkref}[1]{\hyperref[#1]{{Remark~\ref*{#1}}}}
\newcommand{\propref}[1]{\hyperref[#1]{{Proposition~\ref*{#1}}}}
\newcommand{\factref}[1]{\hyperref[#1]{{Fact~\ref*{#1}}}}
\newcommand{\defref}[1]{\hyperref[#1]{{Definition~\ref*{#1}}}}
\newcommand{\exampleref}[1]{\hyperref[#1]{{Example~\ref*{#1}}}}
\newcommand{\hypref}[1]{\hyperref[#1]{{Hypothesis~\ref*{#1}}}}
\newcommand{\secref}[1]{\hyperref[#1]{{Section~\ref*{#1}}}}
\newcommand{\chapref}[1]{\hyperref[#1]{{Chapter~\ref*{#1}}}}
\newcommand{\apref}[1]{\hyperref[#1]{{Appendix~\ref*{#1}}}}

\newcommand\blfootnote[1]{
  \begingroup
  \renewcommand\thefootnote{}\footnote{#1}
  \addtocounter{footnote}{-1}
  \endgroup
}

\AtBeginDocument{}

\begin{document}
\bstctlcite{IEEEexample:BSTcontrol}

\title{Quantum Algorithms for Deep Convolutional Neural Networks} 

%\author{ 
%Iordanis Kerenidis \thanks{
%CNRS, IRIF, Universit\'e Paris Diderot, Paris, France and  
%Centre for Quantum Technologies, National University of Singapore, Singapore. 
%Email: {\tt jkeren@irif.fr}.} 
%\and
%Jonas Landman \thanks{ IRIF, Universit\'e Paris Diderot, Paris, France. Ecole Polytechnique, Palaiseau, France.
%Email: {\tt jonas.landman@polytechnique.edu}.}   
%\and 
%Alessandro Luongo \thanks{Atos Bull, Les Clayes Sous Bois, France. IRIF, Universit\'e Paris Diderot, Paris, France Email: {\tt aluongo@irif.fr}.} 
%\and 
%Anupam Prakash \thanks{CNRS, IRIF, Universit\'e Paris Diderot, Paris, France Email: { \tt anupamprakash1@gmail.com}.}
%}

\author[1]{Iordanis Kerenidis}
\author[1]{Jonas Landman}
\author[1]{Anupam Prakash}
\affil[1]{CNRS, IRIF, Universit\'e Paris Diderot, Paris, France}

\maketitle

\begin{abstract}
Quantum computing is a new computational paradigm that promises applications in several fields, including  machine learning. In the last decade, deep learning, and in particular Convolutional neural networks (CNN), have become essential for applications in signal processing and image recognition \cite{YLeCun, HandbookCNN}. Quantum deep learning, however remains a challenging problem, as it is difficult to implement non linearities with quantum unitaries \cite{schuld2014quest}. In this paper we propose a quantum algorithm for applying and training deep convolutional neural networks with a potential speedup. The quantum CNN (QCNN) is a shallow circuit, reproducing completely the classical CNN, by allowing non linearities and pooling operations. The QCNN is particularly interesting for deep networks and could allow new frontiers in image recognition, by using more or larger convolution kernels, larger or deeper inputs.
We introduce a new quantum tomography algorithm with $\ell_{\infty}$ norm guarantees, and new applications of probabilistic sampling in the context of information processing. 
We also present numerical simulations for the classification of the MNIST dataset to provide practical evidence for the efficiency of the QCNN.
\end{abstract} 

\blfootnote{Email: landman@irif.fr}

\thispagestyle{empty}

\setcounter{page}{1}

\section{Introduction}
The growing importance of deep learning in industry and in our society will require extreme computational power as the dataset sizes and the complexity of these algorithms are expected to increase. Quantum computers are a good candidate to answer this challenge.  The recent progress in the physical realization of quantum processors and the advances in quantum algorithms increase more than ever the need to understand their capabilities and limits. In particular, the field of quantum machine learning has witnessed many innovative algorithms that offer speedups over their classical counterparts \cite{qmeans, LMR13, lloyd2014quantum, KP17, WKS14}. 

Quantum deep learning, the problem of creating quantum circuit that enhance the operations of neural networks, has been studied in several works \cite{kerenidis2018neural, rebentrost2018quantum, wiebe2014quantum}. It however remains a challenging problem as as it is difficult to implement non linearities with quantum unitaries \cite{schuld2014quest}. Convolutional neural networks (CNN) are a type of deep learning architecture well suited for visual recognition, signal processing and time series. In this work we propose a quantum algorithm to perform a complete convolutional neural network (QCNN) that offers potential speedups over classical CNNs. We also provide  results of numerical simulations to evaluate the  running time and accuracy of the QCNN. 

Convolutional neural networks were originally developed by Y. LeCun \cite{YLeCun} in the 1980's. They have achieved many practical successes over the last decade, due to the novel computational power of GPUs. They are now the most widely used algorithms for image recognition tasks \cite{krizhevsky2012imagenet}. Their capacities have been used in various domains such as autonomous cars \cite{bojarski2016visualbackprop} or gravitational wave detection \cite{george2018deep}. CNNs have also opened the path to Generative Adversarial Networks (GANs) which are used to generate or reconstruct images \cite{goodfellow2014generative}. 
Despite these successes, CNNs suffer from a computational bottleneck that make deep CNNs resource expensive in practice.

On the other side, the growing interest on quantum computing has lead researchers to develop different variants of Quantum Neural Networks (QNN). This is a challenging problem for quantum computing due to the modular layer architecture of the neural networks and the presence of non linearities, pooling, and other non unitary operations \cite{schuld2014quest}. Several strategies have been proposed to implement a quantum neural networks \cite{kerenidis2018neural, wiebe2014quantum, beer2019efficient} and even for quantum convolutional neural networks \cite{cong2018quantum} designed to solve specific problems. Another path to design QNN on near term devices is the use of variational quantum circuits \cite{farhi2018classification, henderson2019quanvolutional, killoran2018continuous}. These approaches are interesting in the context of NISQ (Noisy Intermediate Scale Quantum) devices \cite{preskill2018quantum}, however more work is needed in order to provide evidence that such techniques can outperform classical neural networks.

\section{Main results}
In this paper, we design a quantum algorithm for a complete CNN, with a modular architecture that allows any number of layers, any number and size of kernels,  and that includes a large variety of non linearity and pooling methods. We introduce a quantum convolution product and a specific quantum sampling technique well suited for information recovery in the context of CNN. We also propose a quantum algorithm for backpropagation that allows an efficient training of our quantum CNN.

As explained in Section \ref{classicalCNN}, a single layer $\ell$ of the classical CNN does the following operations: from an input image $X^{\ell}$ seen as a 3D tensor, and a kernel $K^{\ell}$ seen as a 4D tensor, it performs a convolution $X^{\ell+1}=X^{\ell} * K^{\ell}$, followed by a non linear function, and followed by pooling. In the quantum case, we will obtain a quantum state corresponding to this output, approximated with error $\epsilon > 0$. To retrieve a classical description from this quantum state, we will apply an $\ell_{\infty}$ tomography (see Theorem \ref{thm:tom}) 
and sample the elements with high values, in order to reduce the computation while still getting the information that matters (see Section \ref{QCNNForward}). 

The QCNN can be directly compared to the classical CNN as it has the same inputs and outputs. 
We show that it offers a speedup compared to the classical CNN for both the forward pass and 
for training using backpropagation in certain cases.  For each layer, on the forward pass (Algorithm \ref{QCNNLayer}), the speedup is exponential in the size of the layer (number of kernels) and almost quadratic on the spatial dimension of the input. We next state informally the speedup for the forward pass, the formal version appears as Theorem \ref{theorem1}. 

\begin{result}\label{r1}{(Quantum Convolution Layer)\\}
Let $X^{\ell}$ be the input and $K^{\ell}$ be the kernel for layer $\ell$ of a convolutional neural network, and $f : \mathbb{R} \mapsto [0,C]$ with $C > 0$ be a non linear function so that 
$f(X^{\ell+1}) := f(X^{\ell}*K^{\ell})$ is the output for layer $\ell$. 

Given $X^{\ell}$ and $K^{\ell}$ stored in QRAM, there is a quantum algorithm that, for any precision parameters $\epsilon > 0$ and $\eta > 0$, creates quantum state $\ket{f(\overline{X}^{\ell+1})}$ 
such that $ \norm{ f(\overline{X}^{\ell+1}) - f(X^{\ell+1}) }_{\infty}  \leq 2M\epsilon$ and retrieves classical tensor $\mathcal{X}^{\ell+1}$ such that for each pixel $j$,

\begin{equation}
	\begin{cases}
	|\mathcal{X}^{\ell+1}_{j} - f(X^{\ell+1}_j)| \leq 2\epsilon \quad \text{if} \quad f(\overline{X}^{\ell+1}_{j}) \geq \eta\\
	\mathcal{X}^{\ell+1}_{j} = 0 \qquad\qquad\qquad \text{if} \quad f(\overline{X}^{\ell+1}_{j}) < \eta\\
	\end{cases}
\end{equation}
This algorithm runs in time 
 \begin{equation}
 \widetilde{O}\left(\frac{1}{\epsilon \eta^2 } \cdot \frac{M \sqrt{C}}{\sqrt{\mathbb{E}(f(\overline{X}^{\ell+1}))}}\right)
 \end{equation}
 where $\mathbb{E}(f(\overline{X}^{\ell+1}))$ represents the average value of $f(\overline{X}^{\ell+1})$, and $\widetilde{O}$ hides factors poly-logarithmic in the size of $X^{\ell}$ and $K^{\ell}$ and the parameter $M$ (defined in Equation (\ref{Mdefinition})) is the maximum product of norms from subregions of $X^{\ell}$ and $K^{\ell}$.

\end{result} 

We see that the number of  kernels, contribute only poly-logarithmic factors to the running time, allowing the QCNN to work with larger and in particular exponentially deeper kernels. The contribution from the input size  is hidden in the precision parameter $\eta$. Indeed, a sufficiently large fraction of pixels must be sampled from the output of the quantum convolution to retrieve the meaningful information. In the Numerical Simulations (Section \ref{NumericalSimulations}) we give estimates for $\eta$. The cost of generating the output $\mathcal{X}^{\ell+1}$ of the quantum convolutional layer 
can thus be much smaller than that for the classical CNN in certain cases, Section \ref{runningtime} provides a detailed comparison to the classical running time.

Following the forward pass, a loss function $\mathcal{L}$ is computed for the output of a classical CNN. The backpropagation algorithm is then used to calculate, layer by layer, the gradient of this loss with respect to the elements of the kernels $K^{\ell}$, in order to update them through gradient descent. The formal version of our quantum backpropagation algorithm is given as Theorem \ref{theorem2}

\begin{result}{(Quantum Backpropagation for Quantum CNN)\\}\label{r2}
Given the forward pass quantum algorithm in Result \ref{r1}, and given the kernel matrix $F^{\ell}$, the input matrices $A^{\ell}$ and $Y^{\ell}$, stored in the QRAM for each layer $\ell$, and a loss function $\mathcal{L}$, there is a quantum backpropagation algorithm that estimates each element of the gradient tensor $\frac{\partial \mathcal{L}}{\partial F^{\ell}}$ within additive error $\delta \norm{ \frac{\partial \mathcal{L}}{\partial F^{\ell}} } $, and updates them to perform a gradient descent. The running time of a single layer $\ell$ for quantum backpropagation is given by
\begin{equation}
O\left(\left(\left(\mu(A^{\ell})+\mu(\frac{\partial \mathcal{L}}{\partial Y^{\ell+1}})\right)\kappa(\frac{\partial \mathcal{L}}{\partial F^{\ell}})+\left(\mu(\frac{\partial \mathcal{L}}{\partial Y^{\ell+1}})+\mu(F^{\ell})\right)\kappa(\frac{\partial \mathcal{L}}{\partial Y^{\ell}})\right) \frac{\log{1/\delta}}{\delta^{2}}\right)
\end{equation}
where for a matrix $V \in \R^{n\times n}$, $\kappa(V)$ is the condition number and $\mu(V)\leq \sqrt{n}$ is a matrix dependent parameter defined in Equation (\ref{mudefinition}).
\end{result}

Details concerning the tensors and their matrix expansion or reshaping are given in Section \ref{tensors}, and a summary of all variables with their meaning and dimension is given in Section \ref{variable_summary_section}. Note that $X^{\ell}$, $Y^{\ell}$ and $A^{\ell}$ are different forms of the same input. Similarly $K^{\ell}$ and $F^{\ell}$ both refer to the kernels.

For the quantum back-propagation algorithm, we introduce a quantum tomography algorithm with $\ell_{\infty}$ norm guarantees, that could be of independent interest. It is exponentially faster than 
tomography with $\ell_2$ norm guarantees and is given as Theorem \ref{thm:tom} in Section \ref{proof:tom}. Numerical simulations on classifying the MNIST dataset show that our quantum CNN achieves a similar classification accuracy as the classical CNN. 

The paper is organized as follows: we first introduce notation and describe the mathematical details for the classical CNN in Section \ref{classicalCNN}. Then, we introduce some necessary quantum preliminaries (Section \ref{quantumpreliminaries}) before explaining our quantum algorithm in two parts: the forward quantum convolution layer (Section \ref{QCNNForward}) and the quantum backpropagation (Section \ref{backpropagation}). The final part presents the results of our numerical simulations (Section \ref{NumericalSimulations}) and our conclusions (Section \ref{conclusions}). A summary of the variables is given as Section \ref{variable_summary_section}, and the two algorithms for the forward and backpropagation phase of the QCNN are given as Algorithm \ref{QCNNLayer} and Algorithm \ref{QBackpropagation}.\\

\section{Classical Convolutional neural network (CNN)}\label{classicalCNN}

CNN is a specific type of neural network, designed in particular for image processing or time series. It uses the \emph{Convolution Product} as a main procedure for each layer. We will focus on image processing with a tensor framework for all elements of the network. Our goal is to explicitly describe the CNN procedures in a form that can be translated in the context of quantum algorithms. 
As a regular neural network, a CNN should learn how to classify any input, in our case images. The training consists of optimizing parameters, learned on the inputs and their corresponding labels. 

\subsection{Tensor representation}
Images, or more generally layers of the network, can be seen as tensors. A tensor is a generalization of a matrix to higher dimensions. For instance an image of height $H$ and width $W$ can be seen as a matrix in $\R^{H\times W}$, where every pixel is a greyscale value between 0 ans 255 (8 bit). However the three channels of color (RGB: Red Green Blue) must be taken into account, by stacking three times the matrix for each color. The whole image is then seen as a 3 dimensional tensor in $\R^{H\times W \times D}$ where $D$ is the number of channels. We will see that the Convolution Product in the CNN can be expressed between 3-tensors (input) and 4-tensors (convolution \emph{filters} or \emph{kernels}), the output being a 3-tensor of different dimensions (spatial size and number of channels).

\begin{figure}[H]
\centering
\includegraphics[height=35mm] {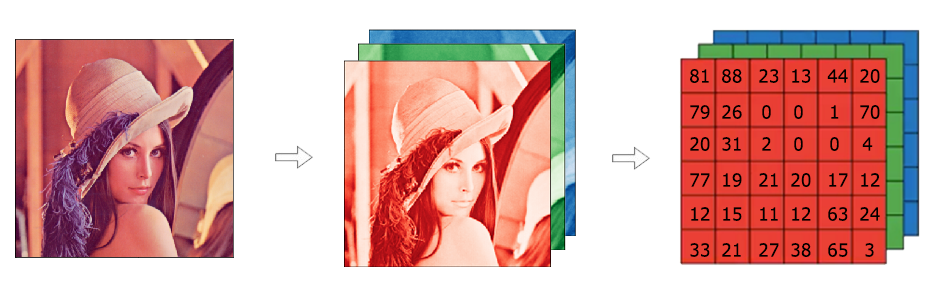} 
\captionsetup{justification=raggedright, margin=1cm}
\caption{RGB decomposition, a colored image is a 3-tensor.}\label{RGB}
\end{figure} 

\subsection{Architecture}

A CNN is composed of 4 main procedures, compiled and repeated in any order : Convolution layers, most often followed by an Activation Function, Pooling Layers and some Fully Connected layers at the end. We will note $\ell$ the current layer.

\paragraph{Convolution Layer :} The $\ell^{th}$ layer is convolved by a set of filters called \emph{kernels}. The output of this operation is the $(\ell+1)^{th}$ layer. A convolution by a single kernel can be seen as a feature detector, that will screen over all regions of the input. If the feature represented by the kernel, for instance a vertical edge, is present in some part of the input, there will be a high value at the corresponding position of the output. The output is called the \emph{feature map} of this convolution.

\paragraph{Activation Function :} As in regular neural network, we insert some non linearities also called \emph{activation functions}. These are mandatory for a neural network to be able to learn any function. In the case of a CNN, each convolution is often followed by a Rectified Linear Unit function, or \emph{ReLu}. This is a simple function that puts all negative values of the output to zero, and lets the positive values as they are. 

\paragraph{Pooling Layer :} This downsampling technique reduces the dimensionality of the layer, in order to improve the computation. Moreover, it gives to the CNN the ability to learn a representation invariant to small translations. Most of the time, we apply a Maximum Pooling or an Average Pooling. The first one consists of replacing a subregion of $P\times P$ elements only by the one with the maximum value. The second does the same by averaging all values. Recall that the value of a pixel corresponds to how much a particular feature was present in the previous convolution layer.  

\paragraph{Fully Connected Layer :} After a certain number of convolution layers, the input has been sufficiently processed so that we can apply a fully connected network. Weights connect each input to each output, where inputs are all element of the previous layer. The last layer should have one node per possible label. Each node value can be interpreted as the probability of the initial image to belong to the corresponding class.

\begin{figure}[!h]
\centering
\includegraphics[height=68mm] {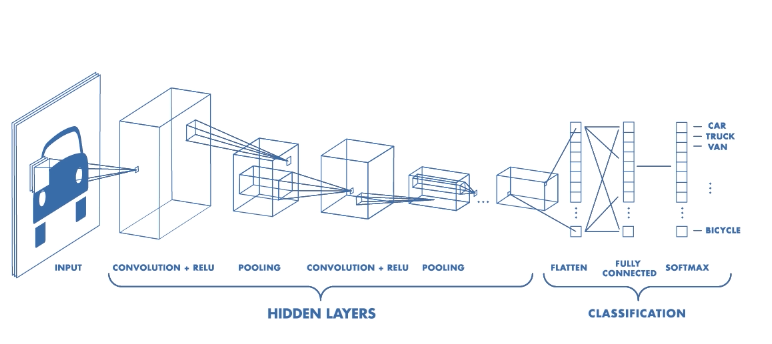} 
\captionsetup{justification=raggedright, margin=1cm}
\caption[Caption for LOF]{Representation of a CNN layers and procedures \footnotemark}
\label{CNN-image}
\end{figure} 
\footnotetext{source: https://fr.mathworks.com/solutions/deep-learning/convolutional-neural-network.html}

\subsection{Convolution Product as a Tensor Operation}\label{tensors}
Most of the following mathematical formulations have been very well detailed by J.Wu in \cite{CNNIntro}. At layer $\ell$, we consider the convolution of a multiple channels image, seen as a 3-tensor $X^{\ell} \in \R^{H^{\ell}\times W^{\ell} \times D^{\ell}}$. Let's consider a single kernel in $\R^{H\times W \times D^{\ell}}$. Note that its third dimension must match the number of channels of the input, as in Figure \ref{volume-convolution}. The kernel passes over all possible regions of the input and outputs a value for each region, stored in the corresponding element of the output. Therefore the output is 2 dimensional, in $\R^{H^{\ell+1}\times W^{\ell+1}}$

\begin{figure}[H]
\centering
\includegraphics[scale = 0.7] {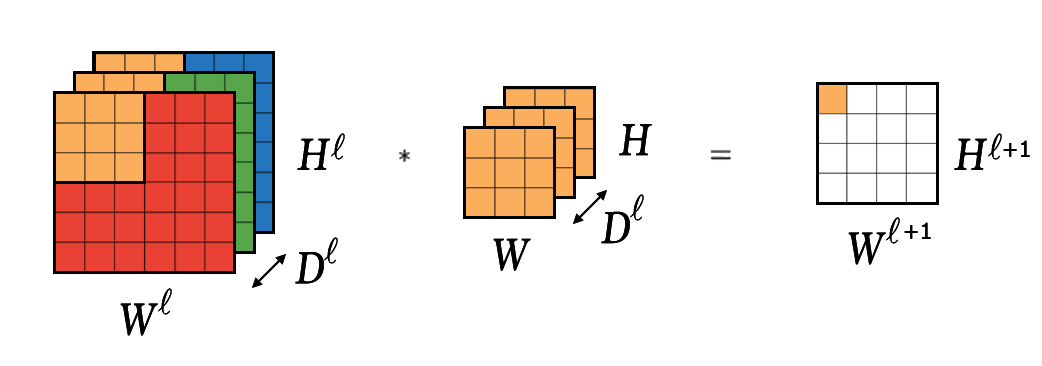} 
\captionsetup{justification=raggedright, margin=1cm}
\caption{Convolution of a 3-tensor input (Left) by one 3-tensor kernel (Center). The ouput (Right) is a matrix for which each entry is a inner product between the kernel and the corresponding overlapping region of the input.}\label{volume-convolution}
\end{figure} 

In a CNN, the most general case is to apply several convolution products to the input, each one with a different 3-tensor kernel. Let's consider an input convolved by $D^{\ell+1}$ kernels. We can globally see this process as a whole, represented by one 4-tensor kernel $K^{\ell} \in \R^{H\times W \times D^{\ell} \times D^{\ell+1}}$. As $D^{\ell+1}$ convolutions are applied, there are $D^{\ell+1}$ outputs of 2 dimensions, equivalent to a 3-tensor $ X^{\ell+1} \in \R^{H^{\ell+1}\times W^{\ell+1} \times D^{\ell+1}}$

\begin{figure}[H]
\centering
\includegraphics[scale = 0.7] {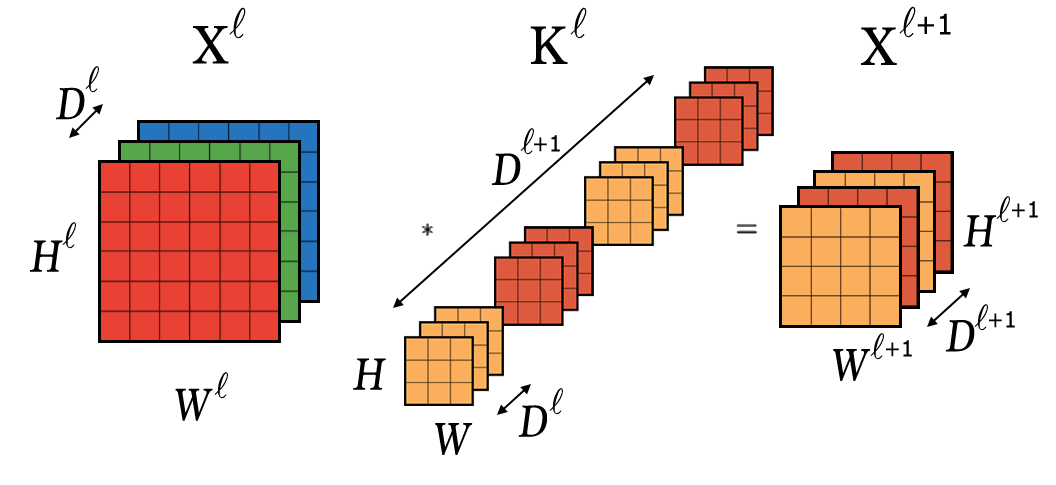} 
\captionsetup{justification=raggedright, margin=1cm}
\caption{Convolutions of the 3-tensor input $X^{\ell}$ (Left) by one 4-tensor kernel $K^{\ell}$ (Center). Each channel of the output $X^{\ell+1}$ (Right) corresponds to the output matrix of the convolution with one of the 3-tensor kernel.}\label{volume-convolution-tensor}
\end{figure} 

This tensor convention explains why Figure \ref{CNN-image} is represented with layers as volumes of different shapes. Indeed we can see on Figure \ref{volume-convolution-tensor} that the output's dimensions are modified given the following rule:

\begin{equation}
  \begin{cases}
    H^{\ell+1} = H^{\ell}-H+1\\
    W^{\ell+1} = W^{\ell}-W+1\\
  \end{cases}
\end{equation}

We omit to detail the use of \emph{Padding} and \emph{Stride}, two parameters that control how the kernel moves through the input, but these can easily be incorporated in the algorithms.

An element of $X^{\ell}$ is determined by 3 indices $(i^{\ell},j^{\ell},d^{\ell})$, while an element of the kernel $K^{\ell}$ is determined by 4 indices $(i,j,d,d')$. For an element of $X^{\ell+1}$ we use 3 indices $(i^{\ell+1},j^{\ell+1},d^{\ell+1})$. We can express the value of each element of the output $X^{\ell+1}$ with the relation

\begin{equation}\label{analytic-expression}
X^{\ell+1}_{i^{\ell+1},j^{\ell+1},d^{\ell+1}} = \sum_{i=0}^{H}\sum_{j=0}^{W}\sum_{d=0}^{D^{\ell}}K^{\ell}_{i,j,d,d^{\ell+1}}X^{\ell}_{i^{\ell+1}+i, j^{\ell+1}+j, d}
\end{equation}

\subsection{Matrix Expression}

\begin{figure}[H]
\centering
\includegraphics[scale=0.65] {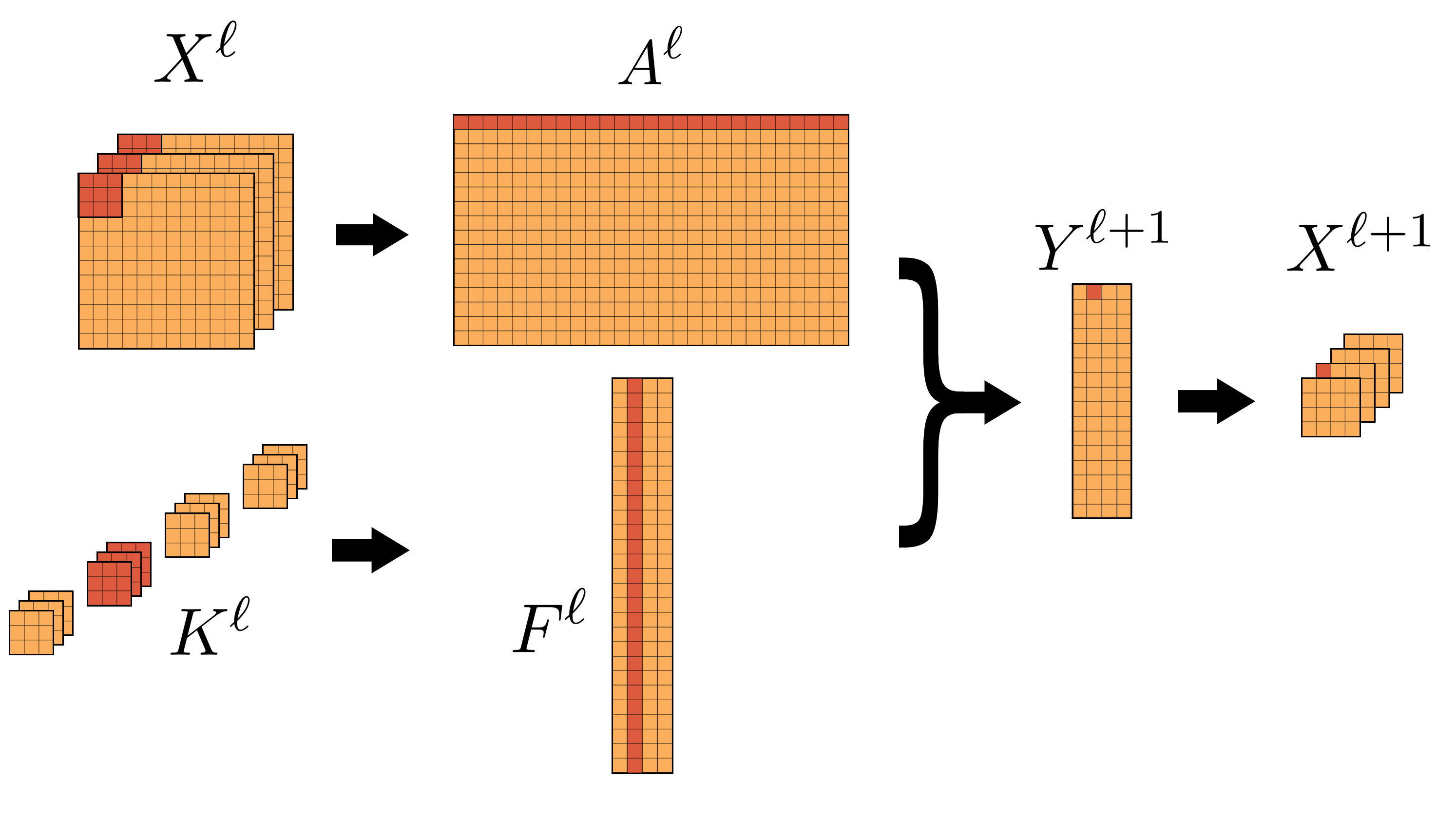} 
\captionsetup{justification=raggedright, margin=1cm}
\caption{A convolution product is equivalent to a matrix-matrix multiplication.}
\label{CNN-tensors-figure}
\end{figure}

It is possible to reformulate Equation (\ref{analytic-expression}) as a matrix product. For this we have to reshape our objects. We expand the input $X^{\ell}$ into a matrix $A^{\ell} \in \R^{(H^{\ell+1}W^{\ell+1})\times(HWD^{\ell})}$. Each row of $A^{\ell}$ is a vectorized version of a subregion of $X^{\ell}$. This subregion is a volume of the same size as a single kernel volume $H\times W\times D^{\ell}$. Hence each of the $H^{\ell+1}\times W^{\ell+1}$ rows of $A^{\ell}$ is used for creating one value in $X^{\ell+1}$. Given such a subregion of $X^{\ell}$,  the rule for creating the row of $A^{\ell}$ is to stack, channel by channel, a column first vectorized form of each matrix. Then, we reshape the kernel tensor $K^{\ell}$ into a matrix $F^{\ell} \in \R^{(HWD^{\ell})\times D^{\ell+1}}$, such that each column of $F^{\ell}$ is a column first vectorized version of one of the $D^{\ell+1}$ kernels.

As proved in \cite{CNNIntro}, the convolution operation $X^{\ell} * K^{\ell} = X^{\ell+1}$ is equivalent to the following matrix multiplication 
\begin{equation}
A^{\ell}F^{\ell} = Y^{\ell+1},
\end{equation}
where each column of $Y^{\ell+1} \in \R^{(H^{\ell+1}W^{\ell+1})\times D^{\ell+1}}$ is a column first vectorized form of one of the $D^{\ell+1}$ channels of $X^{\ell+1}$. Note that an element $Y^{\ell+1}_{p,q}$ is the inner product between the $p^{th}$ row of $A^{\ell}$ and the $q^{th}$ column of $F^{\ell}$. It is then simple to convert $Y^{\ell+1}$ into $X^{\ell+1}$ 
The indices relation between the elements $Y^{\ell+1}_{p,q}$ and $X^{\ell+1}_{i^{\ell+1},j^{\ell+1},d^{\ell+1}}$ is given by:
\begin{equation}\label{YtoX}
  \begin{cases}
    d^{\ell+1} = q\\
    j^{\ell+1} = \floor{\frac{p}{ H^{\ell+1}}}\\
    i^{\ell+1} = p - H^{\ell+1}\floor{\frac{p}{ H^{\ell+1}}}\
  \end{cases}
\end{equation}

A summary of all variables along with their meaning and dimensions is given in Section \ref{variable_summary_section}.

\section{Quantum Preliminaries}\label{quantumpreliminaries}

A single qubit can be $\ket{0}$, $\ket{1}$ or in a superposition state $\alpha\ket{0}+\beta\ket{1}$ with amplitudes $(\alpha,\beta)\in \mathbb{C}$ such that $|\alpha|^2+|\beta|^2=1$. A vector state $\ket{v}$ for $v \in \mathbb{R}^d$ is defined as a quantum superposition on $\ceil{\log(d)}$ qubits $\ket{v} = \frac{1}{\norm{v} } \sum_{i \in [d] } v_i \ket{i}$, where $\ket{i}$ represents the $i^{th}$ vector in the standard basis. Every quantum states will be modified by quantum gates defined as unitary matrices. The measurement of a qubit yields to 0 or 1, each with probability equals to the square of the module of the respective amplitude.  
We will assume at some steps that a matrix $V$ is stored in suitable QRAM data structures which are described in \cite{KP16}. To prove our results, we are going to use the following  tools:

\begin{theorem}[Amplitude estimation \cite{BHMT00}]\label{theoremamplitudeamplification}
	Given a quantum algorithm $$A:\ket{0} \to \sqrt{p}\ket{y,1} + \sqrt{1-p}\ket{G,0}$$ where $\ket{G}$ is some garbage state, then for any positive integer $P$, the amplitude estimation algorithm outputs $\tilde{p}$ $(0 \le \tilde p \le 1)$ such that
	$$
	|\tilde{p}-p|\le 2\pi \frac{\sqrt{p(1-p)}}{P}+\left(\frac{\pi}{P}\right)^2
	$$
	with probability at least $8/\pi^2$. It uses exactly $P$ iterations of the algorithm $A$. 
	If $p=0$ then $\tilde{p}=0$ with certainty, and if $p=1$ and $P$ is even, then $\tilde{p}=1$ with certainty.\label{thm:ampest}
\end{theorem}

\noindent In addition to amplitude estimation, we will make use of a tool developed in \cite{WKS14} to boost the probability of getting a good estimate 
for the inner product required for the quantum convolution algorithm. In high level, we take multiple copies of the estimator from the amplitude estimation procedure, compute the median, and reverse the circuit to get rid of the garbage. Here we provide a theorem with respect to time and not query complexity.

\begin{theorem}[Median Evaluation  \cite{WKS14}]\label{median}
	Let $\mathcal{U}$ be a unitary operation that maps 
	$$\mathcal{U}:\ket{0^{\otimes n}}\mapsto \sqrt{a}\ket{x,1}+\sqrt{1-a} \ket{G,0}$$
	for some $1/2 < a \le 1$ in time $T$. Then there exists a quantum algorithm that, for any $\Delta>0$ and for any $1/2<a_0 \le a$, produces a state $\ket{\Psi}$ such that $\|\ket{\Psi}-\ket{0^{\otimes nL}}\ket{x}\|\le \sqrt{2\Delta}$ for some integer $L$, in time  
	$$
	2T\left\lceil\frac{\ln(1/\Delta)}{2\left(|a_0|-\frac{1}{2} \right)^2}\right\rceil.
	$$
	\label{lem:median}
\end{theorem}

We also need some state preparation procedures. These subroutines are needed for encoding vectors in $v_{i} \in \R^{d}$ into quantum states $\ket{v_{i}}$. An efficient state preparation procedure is provided by the QRAM data structures. 
\begin{theorem}[QRAM data structure \cite{KP16}]\label{qram}
	Let $V \in \mathbb{R}^{n \times d}$, there is a data structure to store the rows of $V$ such that, 
	\begin{enumerate}
		\item The time to insert, update or delete a single entry $v_{ij}$ is $O(\log^{2}(n))$. 	
		\item A quantum algorithm with access to the data structure can perform the following unitaries in time $T=O(\log^{2}n)$. 
		\begin{enumerate} 
			\item $\ket{i}\ket{0} \to \ket{i}\ket{v_{i}} $ for $i \in [n]$. 
			\item $\ket{0} \to \sum_{i \in [n]} \norm{v_{i}}\ket{i}$. 
		\end{enumerate}
	\end{enumerate}
\end{theorem}

In our algorithm we will also use subroutines for quantum linear algebra. 
For a symmetric matrix $M \in \R^{d\times d}$ with spectral norm $\norm{M}$ stored in the QRAM, 
the running time of these algorithms depends linearly on the condition number $\kappa(M)$ of the matrix, that can be replaced by $\kappa_\tau(M)$, a condition threshold where we keep only the singular values bigger than $\tau$, and the parameter $\mu(M)$, a matrix dependent parameter defined as
\begin{equation}\label{mudefinition}
\mu(M)=\min_{p\in [0,1]} \left(\frac{\norm{M}_{F}}{\norm{M}}, \frac{\sqrt{s_{2p}(M)s_{(1-2p)}(M^{T})}}{\norm{M}}\right)
\end{equation}

for $s_{p}(M) = \max_{i \in [n]} \sum_{j \in [d]} M_{ij}^{p}$.
The different terms in the minimum in the definition of $\mu(M)$ correspond to different choices for the data structure for storing $M$, as detailed in \cite{KP17}.  Note that $\mu(M) \leq \frac{\norm{M}_{F}}{{\norm{M}}} \leq \sqrt{d}$.

\begin{theorem}[Quantum linear algebra \cite{CGJ18} ]\label{QuantumLinearAlgebra}  Let $M \in \mathbb{R}^{d \times d}$ and $x \in \mathbb{R}^d$. Let $\delta_1,\delta_2>0$. If $M$ is stored in appropriate QRAM data structures and the time to prepare $\ket{x}$ is $T_{x}$, then there exist quantum algorithms that with probability at least $1-1/poly(d)$ return
    \begin{enumerate}
        \item A state $\ket{z}$ such that $\norm{ \ket{z} - \ket{Mx}}_{2} \leq \delta_1$ in time $\widetilde{O}((\kappa(M)\mu(M) + T_{x} \kappa(M)) \log(1/\delta_1))$.  \\
        Note that this also implies $\norm{ \ket{z} - \ket{Mx}}_{\infty} \leq \delta_1$
        \item Norm estimate $z \in (1 \pm \delta_2)\norm{Mx}_{2}$, with relative error $\delta_2$, in time $\widetilde{O}(T_{x} \frac{\kappa(M)\mu(M)}{\delta_2}\log(1/\delta_1))$.
    \end{enumerate}
    
\end{theorem}
The linear algebra procedures above can also be applied to any rectangular matrix $V \in \mathbb{R}^{n \times d}$ by considering instead the symmetric matrix $ \overline{V} = \left ( \begin{matrix}
0  &V \\ 
V^{T} &0 \\
\end{matrix}  \right )$.  

Finally, we present a logarithmic time algorithm for vector state tomography that will be used to recover classical information from the quantum states with $\ell_{\infty}$ norm guarantee. Given a unitary $U$ that produces a quantum state $\ket{x}$, by calling $O(\log{d}/\delta^{2})$ times $U$, the tomography algorithm is able to reconstruct a vector $\widetilde{X}$ that approximates $\ket{x}$ with $\ell_{\infty}$ norm guarantee, such that $\norm{ \ket{ \widetilde{X}} - \ket{x} }_{\infty} \leq \delta$, or equivalently that $\forall i \in [d], |x_i -  \widetilde{X}_i | \leq \delta$. Such a tomography is of interest when the components $x_i$ of a quantum state are not the coordinates of an meaningful vector in some linear space, but just a series of values, such that we don't want an overall guarantee on the vector (which is the case with usual $\ell_2$ tomography) but a similar error guarantee for each component in the estimation.

\begin{theorem}[$\ell_{\infty}$ Vector state tomography]\label{thm:tom}
Given access to unitary $U$ such that $U\ket{0} = \ket{x}$ and its controlled version in time $T(U)$, there is a tomography algorithm with time complexity $O(T(U) \frac{ \log d  }{\delta^{2}})$ that produces unit vector $\widetilde{X} \in \R^{d}$ such that $\norm{\widetilde{X}  - x }_{\infty} \leq \delta$ with probability at least $(1-1/poly(d))$. 
\end{theorem} 

The proof of this theorem (see Appendix, Section \ref{proof:tom}) is similar to the proof of the $\ell_2$-norm tomography from \cite{KP18}. However the $\ell_{\infty}$ norm tomography introduced in this paper depends only logarithmically and not linearly in the dimension $d$. Note that in our case, $T(U)$ will be logarithmic in the dimension.

\section{Forward pass for the QCNN}\label{QCNNForward}
 
In this section we will design quantum procedures for the usual operations in a CNN layer. We start by describing the main ideas before providing the details. The forward pass algorithm for the QCNN is given as Algorithm \ref{QCNNLayer}. 

First, to perform a convolution product between an input and a kernel, we use the mapping between convolution of tensors and matrix multiplication from Section \ref{tensors}, that can further be reduced to inner product estimation between vectors. The output will be a quantum state representing the result of the convolution product, from which we can sample to retrieve classical information to feed the next layer. This is stated in the following Theorem:

\begin{theorem}\label{theorem1}{(Quantum Convolution Layer) \\}
Given 3D tensor input $X^{\ell} \in \mathbb{R}^{H^{\ell}\times W^{\ell}\times D^{\ell}}$ and 4D tensor kernel $K^{\ell} \in \mathbb{R}^{H\times W\times D^{\ell}\times D^{\ell+1}}$ stored in QRAM, and precision parameters $\epsilon, \Delta >0$, there is a quantum algorithm that computes a quantum states $\Delta$-close to $\ket{f(\overline{X}^{\ell+1})}$ where $X^{\ell+1} = X^{\ell}*K^{\ell}$ and $f : \mathbb{R} \mapsto [0,C]$ is a non linear function. 
A classical approximation such that 
\begin{equation}
\norm{f(\overline{X}^{\ell+1}) - f(X^{\ell+1})}_{\infty} \leq \epsilon
\end{equation}
The time complexity of this procedure is given by
$\widetilde{O}\left( M/\epsilon\right)$, where $M$ is the maximum norm of a product between one of the $D^{\ell+1}$ kernels, and one of the regions of $X^{\ell}$ of size $HWD^{\ell}$ 
and $\widetilde{O}$ hides factors poly-logarithmic in $\Delta$ and in the size of $X^{\ell}$ and $K^{\ell}$.
\end{theorem} 
 Recall that a convolution can be seen as a pattern detection on the input image, where the pattern is the kernel. The output values correspond to ``how much" the pattern was present in the corresponding region of the input. Low value pixels in the output indicate the absence of the pattern in the input at the corresponding regions. Therefore, by sampling according to these output values, where the high value pixels are sampled with more probability, we could retrieve less but only meaningful information for the neural network to learn. It is a singular use case where amplitudes of a quantum state are proportional to the importance of the information they carry, giving a new utility to the probabilistic nature of quantum sampling. Numerical simulations in Section \ref{NumericalSimulations} provide an empirical estimate of the 
 sampling rate to achieve good classification accuracy.

\begin{algorithm} 
\caption{QCNN Layer} \label{QCNNLayer}
\begin{algorithmic}[1]

\REQUIRE  Data input matrix $A^{\ell}$ and kernel matrix $F^{\ell}$ stored in QRAM. Precision parameters $\epsilon$ and $\eta$, a non linearity function $f : \mathbb{R} \mapsto [0,C]$. 
\ENSURE Outputs the data matrix $A^{\ell+1}$ for the next layer, result of the convolution between the input and the kernel, followed by a non linearity and pooling.\\
\vspace{10pt} 

\STATE {\bf Step 1: Quantum Convolution}\\

{\bf 1.1: Inner product estimation}\\
Perform the following mapping, using QRAM queries on rows $A^{\ell}_{p}$ and columns $F^{\ell}_{q}$, along with Theorems \ref{theoremamplitudeamplification} and \ref{median} to obtain
\begin{equation}
\frac{1}{K} \sum_{p,q} \ket{p}\ket{q}
\mapsto 
\frac{1}{K} \sum_{p,q} \ket{p}\ket{q}\ket{\overline{P}_{pq}}\ket{g_{p q}}, 
\end{equation}
where $\overline{P}_{p q}$ is $\epsilon$-close to $P_{p q} = \frac{1+\braket{A^{\ell}_{p}}{F^{\ell}_{q}}}{2}$ and $K = \sqrt{H^{\ell+1}W^{\ell+1}D^{\ell+1}}$ is a normalisation factor. $\ket{g_{p q}}$ is some garbage quantum state. 

{\bf 1.2: Non linearity}\\
Use an arithmetic circuit and two QRAM queries to obtain $\overline{Y}^{\ell+1}$, an $\epsilon$-approximation of the convolution output $Y^{\ell+1}_{p, q} = (A^{\ell}_{p},F^{\ell}_{q})$ and apply the non-linear function $f$ as a boolean circuit to obatin 
\begin{equation}
\frac{1}{K} \sum_{p,q} \ket{p}\ket{q}\ket{f(\overline{Y}^{\ell+1}_{p, q})}\ket{g_{p q}}. 
\end{equation}

\STATE {\bf Step 2: Quantum Sampling} \\
Use Conditional Rotation and Amplitude Amplification to obtain the state 
\begin{equation}
\frac{1}{K} \sum_{p,q} \alpha'_{pq} \ket{p}\ket{q}\ket{f(\overline{Y}^{\ell+1}_{pq})}\ket{g_{p q}}.
\end{equation}
Perform $\ell_{\infty}$ tomography from Theorem \ref{thm:tom} with precision $\eta$, and obtain classically all positions and values $(p,q,f(\overline{Y}^{\ell+1}_{pq}))$ such that, with high probability, values above $\eta$ are known exactly, while others are set to 0.

\STATE {\bf Step 3: QRAM Update and Pooling} \\
Update the QRAM for the next layer $A^{\ell+1}$ while sampling. The implementation of pooling (Max, Average, etc.) can be done by a specific update or the QRAM data structure described in Section \ref{pooling}.

\end{algorithmic}
\end{algorithm}

\subsection{Single Quantum Convolution Layer}\label{singlequantumlayer}
In order to develop a quantum algorithm to perform the convolution as described above, we will make use of quantum linear algebra procedures. We will use quantum states proportional to the rows of $A^{\ell}$, noted $\ket{A_{p}}$, and the columns of $F^{\ell}$, noted $\ket{F_{q}}$ (we omit the $\ell$ exponent in the quantum states to simplify the notation). These states are given by
\begin{equation}
\ket{A_{p}} = \frac{1}{\norm{A_{p}}}\sum_{r=0}^{HWD^{\ell}-1}A_{pr}\ket{r}
\end{equation}
\begin{equation}
\ket{F_{q}} = \frac{1}{\norm{F_{q}}}\sum_{s=0}^{D^{\ell+1}-1}F_{sq}\ket{s}
\end{equation}

\noindent We suppose we can load these vectors in quantum states by performing the following queries:
\begin{equation}
\begin{cases}
\ket{p}\ket{0} \mapsto \ket{p}\ket{A_{p}}\\
\ket{q}\ket{0} \mapsto \ket{q}\ket{F_{q}}
\end{cases}
\end{equation}
Such queries, in time poly-logarithmic in the dimension of the vector, can be implemented with a Quantum Random Access Memory (QRAM). See Section \ref{qramupdate} for more details on the QRAM update rules and its integration layer by layer.\\

\subsubsection{Inner Product Estimation}\label{innerproduct}
The following method to estimate inner products is derived from previous work \cite{qmeans}. With the initial state $\ket{p}\ket{q}\frac{1}{\sqrt{2}}(\ket{0}+\ket{1})\ket{0}$ we apply the queries detailed above in a controlled fashion, followed simply by a Hadamard gate to extract the inner product $\braket{A_{p}}{F_{q}}$ in an amplitude.  
\begin{equation}
\frac{1}{\sqrt{2}}\left(\ket{p}\ket{q}\ket{0}\ket{0}+\ket{p}\ket{q}\ket{1}\ket{0}\right) \mapsto \frac{1}{\sqrt{2}}\left(\ket{p}\ket{q}\ket{0}\ket{A_{p}}+\ket{p}\ket{q}\ket{1}\ket{F_{q}}\right)
\end{equation}
By applying a Hadamard gate on the third register we obtain the following state,
\begin{equation}
\frac{1}{2}\ket{p}\ket{q}\Big(\ket{0}(\ket{A_{p}}+\ket{F_{q}}) + \ket{1}(\ket{A_{p}}-\ket{F_{q}})\Big)
\end{equation}
The probability of measuring $0$ on the third register is given by $P_{p q} = \frac{1+\braket{A_{p}}{F_{q}}}{2}$. Thus we can rewrite the previous state as
\begin{equation}
\ket{p}\ket{q}\Big( \sqrt{P_{p q}}\ket{0, y_{p q}} + \sqrt{1-P_{p q}}\ket{1, y^{'}_{pq}} \Big)
\end{equation}
where $\ket{y_{pq}}$ and $\ket{y'_{pq}}$ are some garbage states. \\
We can perform the previous circuit in superposition. Since $A^{\ell}$ has $H^{\ell+1}W^{\ell+1}$ rows, and $F^{\ell}$ has $D^{\ell+1}$ columns, we obtain the state:
\begin{equation}
\ket{u} =  \frac{1}{\sqrt{H^{\ell+1}W^{\ell+1}D^{\ell+1}}} \sum_{p}\sum_{q} \ket{p}\ket{q}\Big( \sqrt{P_{pq}}\ket{0, y_{pq}} + \sqrt{1-P_{pq}}\ket{1, y^{'}_{pq}} \Big)
\end{equation}
Therefore the probability of measuring the triplet $(p,q,0)$ in the first three registers is given by 
\begin{equation}
P_0(p,q) = \frac{P_{pq}}{H^{\ell+1}W^{\ell+1}D^{\ell+1}} = \frac{1+\braket{A_{p}}{F_{q}}}{2H^{\ell+1}W^{\ell+1}D^{\ell+1}}
\end{equation}
Now we can relate to the Convolution product. Indeed, the triplets $(p,q,0)$ that are the most probable to be measured are the ones for which the value $\braket{A_{p}}{F_{q}}$ is the highest. Recall that each element of $Y^{\ell+1}$ is given by $Y^{\ell+1}_{pq}=(A_{p},F_{q})$, where $``(\cdot,\cdot)"$ denotes the inner product. We see here that we will sample most probably the positions $(p,q)$ for the highest values of $Y^{\ell+1}$, that corresponds to the most important points of $X^{\ell+1}$, by the Equation (\ref{YtoX}). Note that the the values of $Y^{\ell+1}$ can be either positive of negative, which is not an issue thanks to the positiveness of $P_0(p,q)$.

A first approach could be to measure indices $(p,q)$ and rely on the fact that pixels with high values, hence a high amplitude, would have a higher probability to be measured. However we have not exactly the final result, since $\braket{A_{p}}{F_{q}} \neq (A_{p},F_{q}) = \norm{A_{p}}\norm{F_{q}}\braket{A_{p}}{F_{q}}$. Most importantly we then want to apply a non linearity $f(Y^{\ell+1}_{pq})$ to each pixel, for instance the ReLu function, which seems not possible with unitary quantum gates if the data is encoded in the amplitudes only. Morever, due to normalization of the quantum amplitudes and the high dimension of the Hilbert space of the input, the probability of measuring each pixel is roughly the same, making the sampling inefficient. Given these facts, we have added steps to the circuit, in order to measure $(p,q,f(Y^{\ell+1}_{pq}))$, therefore know the value of a pixel when measuring it, while still measuring the most important points in priority. 

\subsubsection{Encoding the amplitude in a register}\label{registerencoding}
Let $\mathcal{U}$ be the unitary that map $\ket{0}$ to $\ket{u}$ 
\begin{equation}
\ket{u} =  \frac{1}{\sqrt{H^{\ell+1}W^{\ell+1}D^{\ell+1}}} \sum_{p,q} \ket{p}\ket{q}\Big( \sqrt{P_{pq}}\ket{0, y_{pq}} + \sqrt{1-P_{pq}}\ket{1, y^{'}_{pq}} \Big)
\end{equation}
The amplitude $\sqrt{P_{pq}}$ can be encoded in an ancillary register by using Amplitude Estimation (Theorem \ref{theoremamplitudeamplification}) followed by a Median Evaluation (Theorem \ref{median}). 

For any $\Delta>0$ and $\epsilon>0$, we can have a state $\Delta$-close to 
\begin{equation}
\ket{u^{'}} =  \frac{1}{\sqrt{H^{\ell+1}W^{\ell+1}D^{\ell+1}}} \sum_{p,q} \ket{p}\ket{q}\ket{0}\ket{\overline{P}_{pq}}\ket{g_{p q}}
\end{equation}
with probability at least $1-2\Delta$, where $|P_{pq} - {\overline{P}_{pq}}| \leq \epsilon$ and $\ket{g_{p q}}$ is a garbage state. This requires $O(\frac{\ln(1/\Delta)}{\epsilon})$ queries of $\mathcal{U}$. In the following we discard the third register $\ket{0}$ for simplicity.

The benefit of having $\overline{P}_{pq}$ in a register is to be able to perform operations on it (arithmetic or even non linear). Therefore we can simply obtain a state corresponding to the exact value of the the convolution product. Since we've built a circuit such that  $P_{pq} = \frac{1+\braket{A_{p}}{F_{q}}}{2}$, with two QRAM calls, we can retrieve the norm of the vectors by applying the following unitary:

\begin{equation}
\ket{p}\ket{q}\ket{\overline{P}_{pq}}\ket{g_{p q}}\ket{0}\ket{0}
\mapsto 
\ket{p}\ket{q}\ket{\overline{P}_{pq}}\ket{g_{p q}}\ket{\norm{A_{p}}}\ket{\norm{F_{q}}}
\end{equation}

On the fourth register, we can then write $Y^{\ell+1}_{pq} = \norm{A_{p}}\norm{F_{q}} \braket{A_{p}}{F_{q}}$ using some arithmetic circuits (addition, multiplication by a scalar, multiplication between registers). We then apply a boolean circuit that implements the ReLu function on the same register, in order to obtain an estimate of $f(Y^{\ell+1}_{pq})$ in the fourth register. We finish by inverting the previous computations and obtain the final state

\begin{equation}\label{afternonlinearity}
\ket{f(\overline{Y}^{\ell+1})} = \frac{1}{\sqrt{H^{\ell+1}W^{\ell+1}D^{\ell+1}}} \sum_{p,q} \ket{p}\ket{q}\ket{f(\overline{Y}^{\ell+1}_{pq})}\ket{g_{p q}}
\end{equation}

Because of the precision $\epsilon$ on $\ket{\overline{P}_{pq}}$, our estimation $\overline{Y}^{\ell+1}_{pq} = (2\overline{P}_{pq}-1)\norm{A_{p}}\norm{F_{q}}$, is obtained with error such that
\begin{equation}\label{errorAE}
|\overline{Y}^{\ell+1}_{pq} - Y^{\ell+1}_{pq}| \leq 2\epsilon \norm{A_{p}}\norm{F_{q}}
\end{equation}

In superposition, we can bound this error by $ |\overline{Y}^{\ell+1}_{pq} - Y^{\ell+1}_{pq}| \leq 2M\epsilon$ where we define 
\begin{equation}\label{Mdefinition}
M = \max_{p,q}{\norm{A_{p}}\norm{F_{q}}}
\end{equation}
$M$ is the maximum product between norms of one of the $D^{\ell+1}$ kernels, and one of the regions of $X^{\ell}$ of size $HWD^{\ell}$. Finally, since Equation (\ref{errorAE}) is valid for all pairs $(p,q)$, the overall error committed on the convolution product can be bounded by $\norm{\overline{Y}^{\ell+1}-Y^{\ell+1}}_{\infty} \leq 2M\epsilon$, where $\norm{.}_{\infty}$ denotes the $\ell_{\infty}$ norm. Recall that $Y^{\ell+1}$ is just a reshaped version of $X^{\ell+1}$. Since the non linearity adds no approximation, we can conclude on the final error committed for a layer of our QCNN
\begin{equation}\label{errorAEfinal}
\norm{f(\overline{X}^{\ell+1})-f(X^{\ell+1})}_{\infty} \leq 2M\epsilon
\end{equation}

At this point, we have established Theorem \ref{theorem1} as we have created the quantum state (\ref{afternonlinearity}), with given precision guarantees, in time poly-logarithmic in $\Delta$ and in the size of $X^{\ell}$ and $K^{\ell}$.

We know aim to retrieve classical information from this quantum state. Note that $\ket{Y^{\ell+1}_{pq}}$ is representing a scalar encoded in as many qubits as needed for the precision, whereas $\ket{A_{p}}$ was representing a vector as a quantum state in superposition, where each element $A_{p,r}$ is encoded in one amplitude (See Section \ref{quantumpreliminaries}). The next step can be seen as a way to retrieve both encoding at the same time, that will allow an efficient tomography focus on the values of high magnitude.

\subsubsection{Conditional rotation}
In the following sections, we omit the $\ell+1$ exponent for simplicity.
Garbage states are removed as they will not perturb the final measurement. We now aim to modify the amplitudes, such that the highest values of $\ket{f(\overline{Y})}$ are measured with higher probability. A way to do so consists in applying a conditional rotation on an ancillary qubit, proportionally to $f(\overline{Y}_{pq})$. We will detail the calculation since in the general case $f(\overline{Y}_{pq})$ can be greater than 1. To simplify the notation, we note $x=f(\overline{Y}_{pq})$. 

This step consists in applying the following rotation on a ancillary qubit:
\begin{equation}
\ket{x}\ket{0} \mapsto \ket{x}\left(\sqrt{\frac{x}{\max{x}}}\ket{0}+\beta \ket{1}\right)
\end{equation}
Where $\max{x} = \max_{p,q}{f(\overline{Y}_{pq})}$ and $\beta = \sqrt{1-(\frac{x}{\max{x}})^2}$. Note that in practice it is not possible to have access to $\ket{\max{x}}$ from the state (\ref{afternonlinearity}), but we will present a method to know \emph{a priori} this value or an upper bound in section \ref{capReLu}.

Let's note $\alpha_{pq} = \sqrt{\frac{f(\overline{Y}_{pq})}{\\max_{p,q}(f(\overline{Y}_{pq}))}}$. The ouput of this conditional rotation in superposition on state (\ref{afternonlinearity}) is then 
\begin{equation}
\frac{1}{\sqrt{HWD}} \sum_{p,q} \ket{p}\ket{q}\ket{f(\overline{Y}_{pq})}(\alpha_{pq}\ket{0}+\sqrt{1-\alpha_{pq}^2}\ket{1})
\end{equation}

\subsubsection{Amplitude Amplification}\label{amplitudeamplification}
In order to measure $(p,q,f(\overline{Y}_{pq}))$ with higher probability where $f(\overline{Y}_{pq})$ has high value, we could post select on the measurement of $\ket{0}$ on the last register. Otherwise, we can perform an amplitude amplification on this ancillary qubit. Let's rewrite the previous state as
\begin{equation}
\frac{1}{\sqrt{HWD}} \sum_{p,q} \alpha_{pq}\ket{p}\ket{q}\ket{f(\overline{Y}_{pq})}\ket{0}+\sqrt{1-\alpha_{pq}^2}\ket{g'_{p q}}\ket{1}
\end{equation}
Where $\ket{g'_{p q}}$ is another garbage state. The overall probability of measuring $\ket{0}$ on the last register is $P(0) = \frac{1}{HWD}\sum_{pq}|\alpha_{pq}|^2$. The number of queries required \cite{BHMT00} to amplify the state $\ket{0}$ is $O(\frac{1}{\sqrt{P(0)}})$. Since $f(\overline{Y}_{pq}) \in \mathbb{R}^{+}$, we have $\alpha_{pq}^2 = \frac{f(\overline{Y}_{pq})}{\max_{p,q}(f(\overline{Y}_{pq}))}$. Therefore the number of queries is 
\begin{equation}
O\left(\sqrt{\max_{p,q}(f(\overline{Y}_{pq}))}\frac{1}{\sqrt{\frac{1}{HWD}\sum_{p,q}f(\overline{Y}_{pq})}}\right) = 
O\left(\frac{\sqrt{\max_{p,q}(f(\overline{Y}_{pq}))}}{\sqrt{\mathbb{E}_{p,q}(f(\overline{Y}_{pq}))}}\right)
\end{equation}
Where the notation $\mathbb{E}_{p,q}(f(\overline{Y}_{pq}))$ represents the average value of the matrix $f(\overline{Y})$. It can also be written $\mathbb{E}(f(\overline{X}))$ as in Result \ref{r1}:
\begin{equation}\label{averagevalue}
\mathbb{E}_{p,q}(f(\overline{Y}_{pq})) = \frac{1}{HWD}\sum_{p,q}f(\overline{Y}_{pq})
\end{equation}
At the end of these iterations, we have modified the state to the following: 

\begin{equation}\label{afteramplitudeencoding}
\ket{f(\overline{Y})} = \frac{1}{\sqrt{HWD}} \sum_{p,q} \alpha'_{pq}\ket{p}\ket{q}\ket{f(\overline{Y}_{pq})}
\end{equation}

Where, to respect the normalization of the quantum state, $\alpha'_{pq} = \frac{\alpha_{pq}}{\sqrt{\sum_{p,q}\frac{\alpha^2_{pq}}{HWD}}}$. Eventually, the probability of measuring $(p,q,f(\overline{Y}_{pq}))$ is given by
\begin{equation}
p(p,q,f(\overline{Y}_{pq})) = \frac{(\alpha'_{pq})^2}{HWD} = \frac{(\alpha_{pq})^2}{\sum_{p,q}(\alpha_{pq})^2}=\frac{f(\overline{Y}_{pq})}{\sum_{p,q}{f(\overline{Y}_{pq})}}
\end{equation}

Note that we have used the same type of name $\ket{f(\overline{Y})}$ for both state (\ref{afternonlinearity}) and state (\ref{afteramplitudeencoding}). For now on, this state name will refer only to the latter (\ref{afteramplitudeencoding}). \\

\subsubsection{$\ell_{\infty}$ tomography and probabilistic sampling}\label{tomographyconvolution}
We can rewrite the final quantum state obtained in (\ref{afteramplitudeencoding}) as 
\begin{equation}\label{finalstate}
\ket{f(\overline{Y}^{\ell+1})} = \frac{1}{\sqrt{\sum_{p,q}{f(\overline{Y}^{\ell+1}_{pq})}}} \sum_{p,q} \sqrt{f(\overline{Y}^{\ell+1}_{pq})} \ket{p}\ket{q}\ket{f(\overline{Y}^{\ell+1}_{pq})}
\end{equation}

We see here that $f(\overline{Y}^{\ell+1}_{pq})$, the values of each pixel, are encoded in both the last register and in the amplitude. We will use this property to extract efficiently the exact values of high magnitude pixels. 
For simplicity, we will use instead the notation $f(\overline{X}^{\ell+1}_{n})$ to denote a pixel's value, with $n \in [H^{\ell+1}W^{\ell+1}D^{\ell+1}]$. Recall that $Y^{\ell+1}$ and $X^{\ell+1}$ are reshaped version of the same object.  

The pixels with high values will have more probability of being sampled. Specifically, we perform a tomography with $\ell_{\infty}$ guarantee and precision parameter $\eta > 0$. See Theorem \ref{thm:tom} and Section \ref{proof:tom} for details. The $\ell_{\infty}$ guarantee allows to obtain each pixel with error at most $\eta$, and require $\widetilde{O}(1/\eta^2)$ samples from the state (\ref{finalstate}). 
Pixels with low values $f(\overline{X}^{\ell+1}_{n}) < \eta$ will probably not be sampled due to their low amplitude. Therefore the error committed will be significative and we adopt the rule of setting them to 0.
Pixels with higher values $f(\overline{X}^{\ell+1}_{n}) \geq \eta$, will be sample with high probability, and only one appearance is enough to get the exact register value $f(\overline{X}^{\ell+1}_{n})$ of the pixel, as is it also written in the last register. 

To conclude, let's note $\mathcal{X}^{\ell+1}_{n}$ the resulting pixel values after the tomography, and compare it to the real classical outputs $f(X^{\ell+1}_{n})$. Recall that the measured values $f(\overline{X}^{\ell+1}_{n})$ are approximated with error at most $2M\epsilon$ with $M = \max_{p,q}{\norm{A_{p}}\norm{F_{q}}}$. The algorithm described above implements the following rules:

\begin{equation}
	\begin{cases}
	|\mathcal{X}^{\ell+1}_{n} - f(X^{\ell+1}_n)| \leq 2M\epsilon \quad \text{if} \quad f(\overline{X}^{\ell+1}_{n}) \geq \eta\\
	\mathcal{X}^{\ell+1}_{n} = 0 \qquad\qquad\qquad\qquad \text{if} \quad f(\overline{X}^{\ell+1}_{n}) < \eta\\
	\end{cases}
\end{equation}

Concerning the running time, one could ask what values of $\eta$ are sufficient to obtain enough meaningful pixels. Obviously this highly depends on the output's size $H^{\ell +1}W^{\ell +1}D^{\ell+1}$ and on the output's content itself. But we can view this question from an other perspective, by considering that we sample a constant fraction of pixels given by $\sigma \cdot (H^{\ell +1}W^{\ell +1}D^{\ell+1})$ where $\sigma \in [0,1]$ is a sampling ratio. Because of the particular amplitudes of state (\ref{finalstate}), the high value pixels will be measured and known with higher probability. The points that are not sampled are being set to 0. We see that this approach is equivalent to the $\ell_{\infty}$ tomography, therefore we have 
\begin{equation}\label{simga_eta_equivalence}
\frac{1}{\eta^2} = \sigma \cdot H^{\ell +1}W^{\ell +1}D^{\ell+1}
\end{equation}
\begin{equation}
\eta = \frac{1}{\sqrt{\sigma \cdot H^{\ell +1}W^{\ell +1}D^{\ell+1}}}
\end{equation}

We will use this analogy in the numerical simulations (Section \ref{NumericalSimulations}) to estimate, for a particular QCNN architecture and a particular dataset of images, which values of $\sigma$ are enough to allow the neural network to learn.

\subsubsection{Regularization of the Non Linearity}\label{capReLu}
In the previous steps, we see several appearances of the parameter $\max_{p,q}(f(\overline{Y}^{\ell+1}_{pq}))$.  First for the conditional rotation preprocessing, we need to know this value or an upper bound. Then for the running time, we would like to bound this parameter. Both problems can be solved by replacing the usual ReLu non linearity by a particular activation function, that we note $capReLu$. This function is simply a parametrized ReLu function with an upper threshold, the cap $C$, after which the function remain constant. The choice of $C$ will be tuned for each particular QCNN, as a tradeoff between accuracy and speed. Otherwise, the only other requirement of the QCNN activation function would be not to allow negative values. This is already often the case for most of the classical CNN. In practice, we expect the capReLu to be as good as a usual ReLu, for convenient values of the cap $C$ ($\leq10$). We performed numerical simulations to compare the learning curve of the same CNN with several values of $C$. See the numerical experiments presented in Section \ref{NumericalSimulations} for more details.

\begin{figure}[h]
\minipage{0.5\textwidth}
	\centering
	\includegraphics[width=60mm] {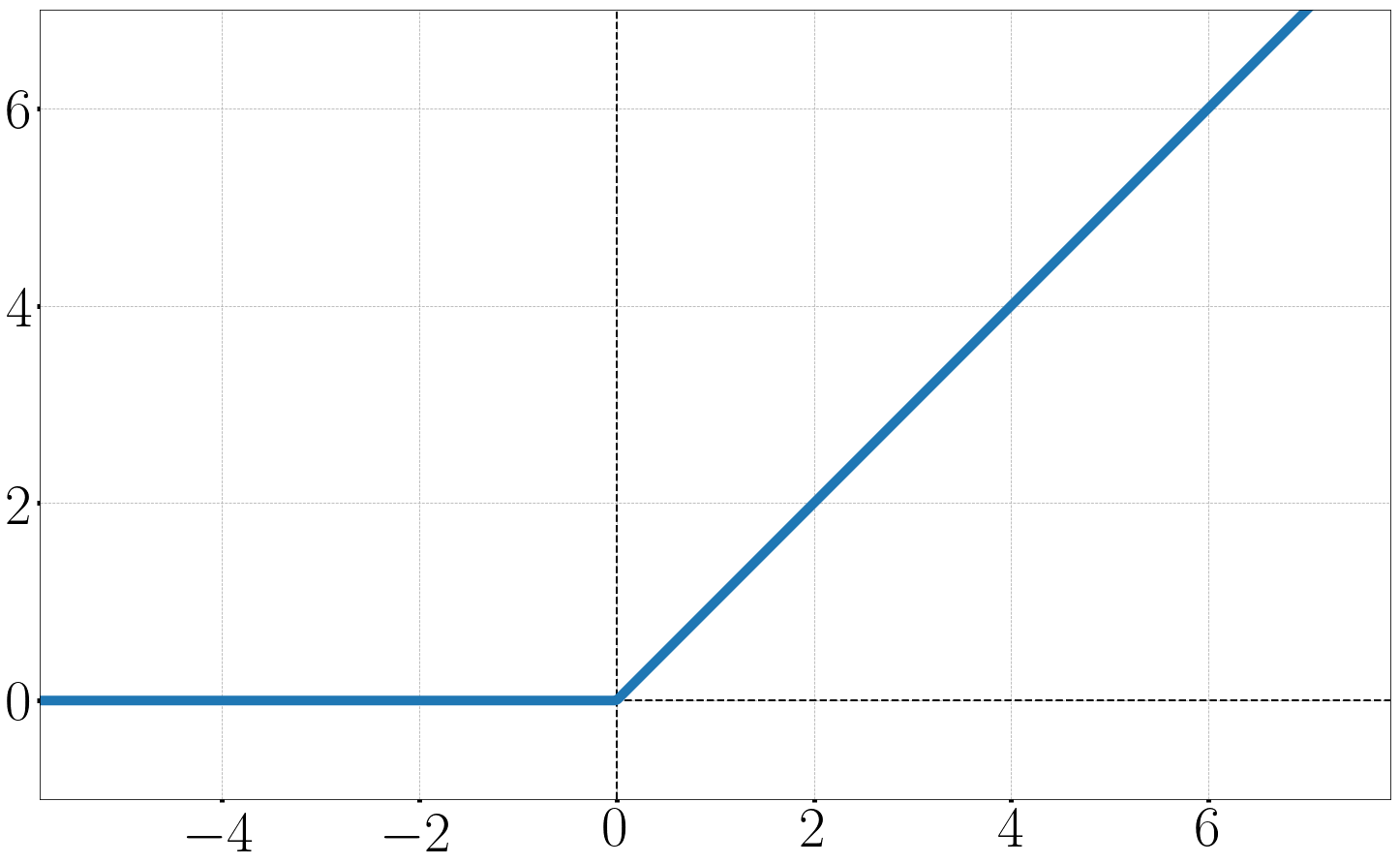} 
\endminipage\hfill
\minipage{0.5\textwidth}
	\centering
	\includegraphics[width=60mm] {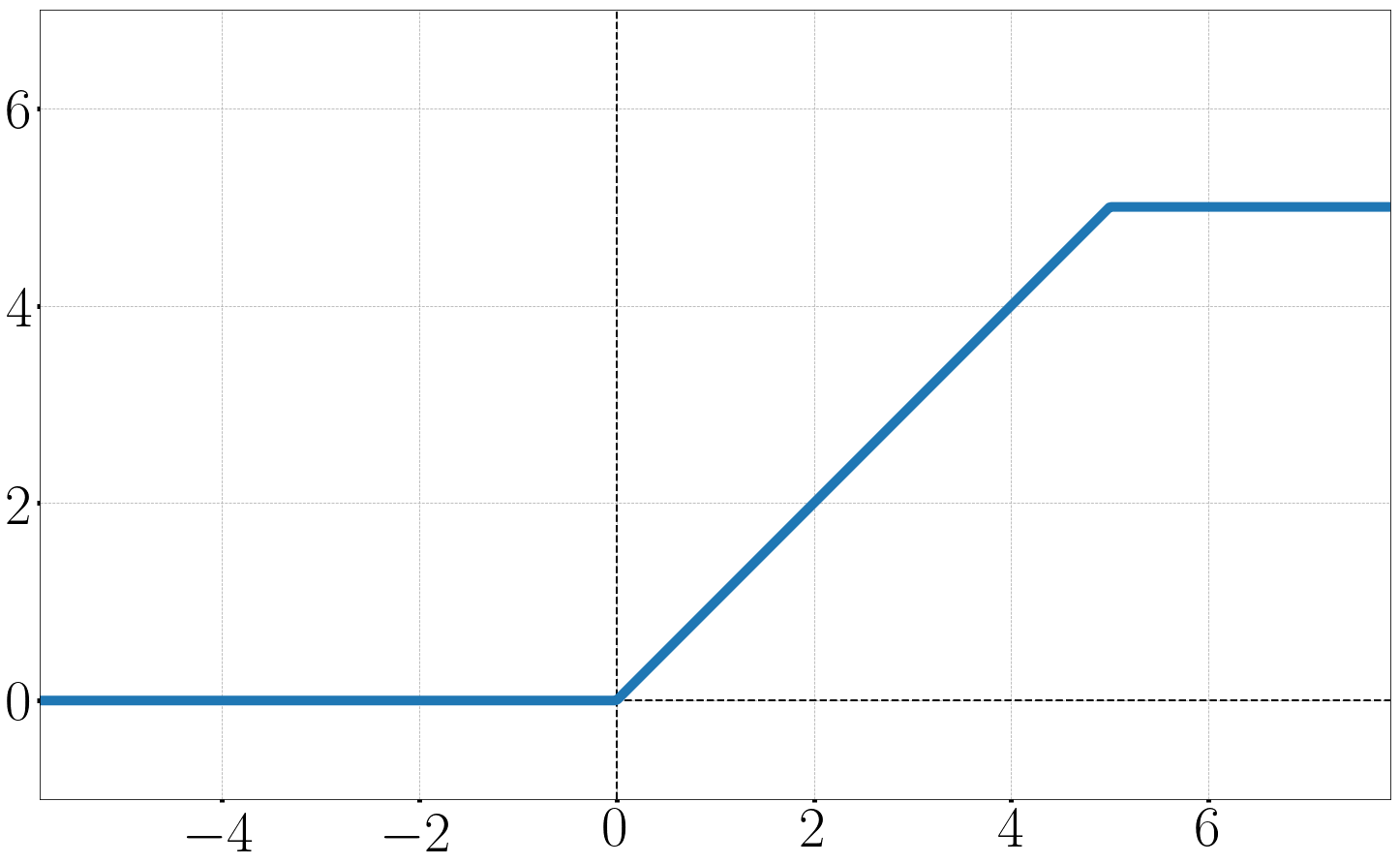} 
\endminipage
\captionsetup{justification=raggedright, margin=1cm}
\caption{Activation functions: ReLu (Left) and capReLu (Right) with a cap $C$ at 5.}
\label{capReLu_5}
\end{figure} 

\subsection{QRAM update}\label{qramupdate}

We wish to detail the use of the QRAM between each quantum convolution layer, and present how the pooling operation can happen during this phase. General results about the QRAM is given as Theorem \ref{qram}. Implementation details can be found in \cite{KP16}. In this section, we will show how to store samples from the output of the layer $\ell$, to create the input of layer $\ell+1$.

\subsubsection{Storing the output values during the sampling}

At the beginning of layer $\ell+1$, the QRAM must store $A^{\ell+1}$, a matrix where each elements is indexed by $(p',r')$, and perform $\ket{p'}\ket{0} \mapsto \ket{p'}\ket{A^{\ell+1}_{p'}}$. The data is stored in the QRAM as a tree structure described in \cite{KP17}. Each row $A^{\ell+1}_{p'}$ is stored in such a tree $T^{\ell+1}_{p'}$. Each leaf $A^{\ell+1}_{p'r'}$ correspond to a value sampled from the previous quantum state $\ket{f(\overline{Y}^{\ell+1})}$, output of the layer $\ell$. The question is to know where to store a sample from $\ket{f(\overline{Y}^{\ell+1})}$ in the tree $T^{\ell+1}_{p'}$.

When a point is sampled from the final state of the quantum convolution, at layer $\ell$, as described in Section \ref{amplitudeamplification}, we obtain a triplet corresponding to the two positions and the value of a point in the matrix $f(\overline{Y}^{\ell+1})$. 
We can know where this point belong in the input of layer $\ell+1$, the tensor $X^{\ell+1}$, by Equation (\ref{YtoX}), since $Y^{\ell}$ is a reshaped version of $X^{\ell}$.

The position in $X^{\ell+1}$, noted $(i^{\ell+1},j^{\ell+1},d^{\ell+1})$, is then matched to several positions $(p',r')$ in $A^{\ell+1}$. For each $p'$, we write in the tree $T^{\ell+1}_{p'}$ the sampled value at leaf $r'$ and update its parent nodes, as required by \cite{KP17}. Note that leaves that weren't updated will be considered as zeros, corresponding to pixels with too low values, or not selected during pooling (see next section). 

Having stored pixels in this way, we can then query $\ket{p'}\ket{0} \mapsto \ket{p'}\ket{A^{\ell}_{p'}}$, using the quantum circuit developed in \cite{KP17}, where we correctly have $\ket{A^{\ell+1}_{p'}} = \frac{1}{\norm{A^{\ell+1}_{p'}}}\sum_{r'}A^{\ell+1}_{p'r'}\ket{r'}$. Note that each tree has a logarithmic depth in the number of leaves, hence the running time of writing the output of the quantum convolution layer in the QRAM gives a marginal multiplicative increase, poly-logarithmic in the number of points sampled from $\ket{f(\overline{Y}^{\ell+1})}$, namely $O(\log(1/\eta^2))$.

\subsubsection{Quantum Pooling}\label{pooling}
As for the classical CNN, a QCNN should be able to perform pooling operations. We first detail the notations for classical pooling. At the end of layer $\ell$, we wish to apply a pooling operation of size P on the output $f(X^{\ell+1})$. We note $\tilde{X}^{\ell+1}$ the tensor after the pooling operation. For a point in $f(X^{\ell+1})$ at position $(i^{\ell+1},j^{\ell+1},d^{\ell+1})$, we know to which \emph{pooling region} it belongs, corresponding to a position $(\tilde{i}^{\ell+1},\tilde{j}^{\ell+1},\tilde{d}^{\ell+1})$ in $\tilde{X}^{\ell+1}$:

\begin{equation}
\begin{cases}
    \tilde{d}^{\ell+1} = d^{\ell+1} \\
    \tilde{j}^{\ell+1} = \floor{\frac{j^{\ell+1}}{P}}\\
    \tilde{i}^{\ell+1} = \floor{\frac{i^{\ell+1}}{P}}
\end{cases}
\end{equation}

\begin{figure}[h]
\centering
\includegraphics[scale=0.8] {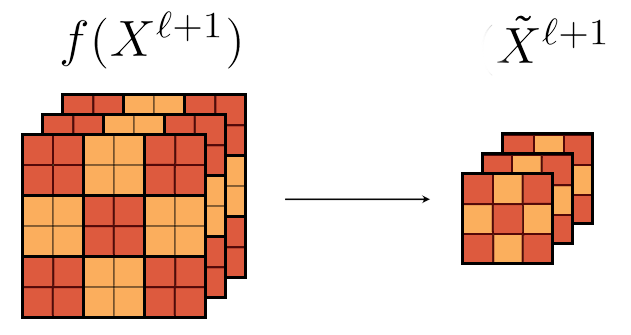} 
\captionsetup{justification=raggedright, margin=1cm}
\caption{A $2\times 2$ tensor pooling. A point in $f(X^{\ell+1})$ (left) is given by its position $(i^{\ell+1},j^{\ell+1},d^{\ell+1})$. A point in $\tilde{X}^{\ell+1}$ (right) is given by its position $(\tilde{i}^{\ell+1},\tilde{j}^{\ell+1},\tilde{d}^{\ell+1})$. Different \emph{pooling regions} in $f(X^{\ell+1})$ have separate colours, and each one corresponds to a unique point in $\tilde{X}^{\ell+1}$.}\label{max-pool}
\end{figure}

We now show how any kind of pooling can be efficiently integrated to our QCNN structure. Indeed the pooling operation will occur during the QRAM update described above, at the end of a convolution layer. At this moment we will store sampled values according to the pooling rules. 

In the quantum setting, the output of layer $\ell$ after tomography is noted $\mathcal{X}^{\ell+1}$. After pooling, we will describe it by $\mathcal{\tilde{X}}^{\ell+1}$, which has dimensions $\frac{H^{\ell+1}}{P} \times \frac{W^{\ell+1}}{P} \times D^{\ell+1}$. $\mathcal{\tilde{X}}^{\ell+1}$ will be effectively used as input for layer $\ell+1$ and its values should be stored in the QRAM to form the trees $\tilde{T}^{\ell+1}_{p'}$, related to the matrix expansion $\tilde{A}^{\ell+1}$. 

However $\mathcal{X}^{\ell+1}$ is not known before the tomography is over. Therefore we have to modify the update rule of the QRAM to implement the pooling in an online fashion, each time a sample from $\ket{f(\overline{X}^{\ell+1})}$ is drawn. Since several sampled values of $\ket{f(\overline{X}^{\ell+1})}$ can correspond to the same leaf $\tilde{A}^{\ell+1}_{p'r'}$ (points in the same \emph{pooling region}), we need an overwrite rule, that will depend on the type of pooling. In the case of Maximum Pooling, we simply update the leaf and the parent nodes if the new sampled value is higher that the one already written. In the case of Average Polling, we replace the actual value by the new averaged value. 

In the end, any pooling can be included in the already existing QRAM update. In the worst case, the running time is increased by $\widetilde{O}(P/\eta^2)$, an overhead corresponding to the number of times we need to overwrite existing leaves, with $P$ being a small constant in most cases. 

As we will see in Section \ref{backpropagation}, the final positions $(p, q)$ that were sampled from $\ket{f(\overline{X}^{\ell+1})}$ and selected after pooling must be stored for further use during the backpropagation phase.

\subsection{Running Time}\label{runningtime}
We will now summarize the running time for one forward pass of convolution layer $\ell$. With $\tilde{O}$ we hide the polylogarithmic factors. We first write the running time of the classical CNN layer, which is given by
\begin{equation}
\widetilde{O}\left(H^{\ell +1}W^{\ell +1}D^{\ell+1} \cdot HWD^{\ell}\right)
\end{equation}

\noindent For the QCNN, the previous steps prove Result \ref{r1} and can be implemented in time
\begin{equation}
\widetilde{O}\left(\frac{1}{\epsilon \eta^2}\cdot \frac{M\sqrt{C}}{\sqrt{\mathbb{E}(f(\overline{X}^{\ell+1}))}} \right)
\end{equation}

\noindent Note that, as explain in Section \ref{tomographyconvolution}, the quantum running time can also be written 
\begin{equation}\label{final_runtime_2}
\widetilde{O}\left(\sigma H^{\ell +1}W^{\ell +1}D^{\ell+1} \cdot \frac{M\sqrt{C}}{\epsilon \sqrt{\mathbb{E}(f(\overline{X}^{\ell+1}))}} \right)
\end{equation}
 with $\sigma \in [0,1]$ being the fraction of sampled elements among $H^{\ell +1}W^{\ell +1}D^{\ell+1}$ of them.

It is interesting to notice that the one quantum convolution layer can also include the ReLu operation and the Pooling operation in the same circuit, for no significant increase in the running time, whereas in the classical CNN each operation must be done on the whole data again. 

Let's go through all the important parameter that appear in the quantum running time:

- The error $\epsilon$ committed during the inner product estimation, is an empirical parameter
Our simulations tend to show that this error can be high without compromising the learning. Indeed the introduction of noise is sometimes interesting in machine learning applications, providing more robust learning \cite{goodfellow2016deep, bishop1995training}.

- The parameter $M = \max_{p,q}{\norm{A_{p}}\norm{F_{q}}}$ as a worst case upper bound during inner product estimation.

- Precision parameter $\eta$ can be related to the fraction of sampled elements in the quantum output $\ket{f(\overline{X}^{\ell+1})}$ of the convolution layer, during $\ell_{\infty}$ tomography.  

- Amplitude amplification adds a multiplicative term $\sqrt{\max{(f(\overline{X}^{\ell+1}))}}$ to the running time, replaced here by $\sqrt{C}$, a constant parameter of order $O(1)$, corresponding to the \emph{cap}, or upper bound, of the activation function. See Section \ref{capReLu} for details. This parameter appears at the conditional rotation step.

- Similarly, the data related value $\mathbb{E}(f(\overline{X}^{\ell+1}))$, appearing during amplitude amplification, denotes the average value in the tensor $f(\overline{X}^{\ell+1})$, as defined in Equation (\ref{averagevalue}).

Finally, in most case, in order to recognize kernel features in the input tensor, the size $H \times W$ of the kernels is a sufficient constant fraction of the input size $H^{\ell} \times W^{\ell}$. Since $H^{\ell+1} = H^{\ell} - H  + 1$, the classical running time can be seen as quadratic in the input size, whereas the quantum algorithm is almost linear. 

\subsection{Variable Summary}\label{variable_summary_section}

We recall the most important variables for layer $\ell$. They represent tensors, their approximations, and their reshaped versions.

\begin{table}[H]
\centering
\begin{tabular}{|c|c|c|c|}
\hline
Data                                & Variable             & Dimensions                                                           & Indices                                              \\ \hline
\multirow{3}{*}{Input}       & $X^{\ell}$           & $H^{\ell}\times W^{\ell}\times D^{\ell}$                             & $(i^{\ell},j^{\ell},d^{\ell})$                       \\ \cline{2-4} 
                                       & $Y^{\ell}$           & $(H^{\ell}W^{\ell})\times D^{\ell}$                                  & -                                                    \\ \cline{2-4} 
                                       & $A^{\ell}$           & $(H^{\ell+1}W^{\ell+1})\times(HWD^{\ell})$                           & $(p,r)$                                              \\ \hline
\multirow{2}{*}{Kernel}    & $K^{\ell}$           & $H\times W \times D^{\ell}\times D^{\ell+1}$                         & $(i,j,d,d')$                                         \\ \cline{2-4} 
                                       & $F^{\ell}$           & $(HWD^{\ell})\times D^{\ell+1}$                                      & $(s,q)$                                             \\ \hline
\end{tabular}
\caption{Summary of input variables for the $\ell^{th}$ layer, along with their meaning, dimensions and corresponding notations. These variables are common for both \emph{quantum} and \emph{classical} algorithms. We have omitted indices for $Y^{\ell}$ which don't appear in our work.}
\label{variable_summary_1}
\end{table}

\begin{table}[H]
\centering
\begin{tabular}{|c|c|c|c|}
\hline
Data                                   & Variable             & Dimensions                                                           & Indices                                              \\ \hline

\multirow{2}{*}{Output of Quantum Convolution} 
				     & $f(\overline{Y}^{\ell+1})$         & $(H^{\ell+1}W^{\ell+1})\times D^{\ell+1}$                            & $(p,q)$                               \\ \cline{2-4} 
				     & $f(\overline{X}^{\ell+1})$         & $H^{\ell+1}\times W^{\ell+1}\times D^{\ell+1}$                    & $(i^{\ell+1},j^{\ell+1},d^{\ell+1})$                    \\ \hline

Output of Quantum Tomography   
				     & $\mathcal{X}^{\ell+1}$ & $H^{\ell+1}\times W^{\ell+1}\times D^{\ell+1}$ & 				$(i^{\ell+1},j^{\ell+1},d^{\ell+1})$        		     \\ \hline

Output of Quantum Pooling           & $\tilde{\mathcal{X}}^{\ell+1}$ & $\frac{H^{\ell+1}}{P} \times \frac{W^{\ell+1}}{P} \times D^{\ell+1}$ &   $(\tilde{i}^{\ell+1},\tilde{j}^{\ell+1},\tilde{d}^{\ell+1})$ 		      \\ \hline

\end{tabular}
\caption{Summary of variables describing outputs of the layer $\ell$, with the \emph{quantum} algorithm.}
\label{variable_summary_2}
\end{table}

\begin{table}[H]
\centering
\begin{tabular}{|c|c|c|c|}
\hline
Data                                   & Variable             & Dimensions                                                           & Indices                                              \\ \hline

\multirow{2}{*}{Output of Classical Convolution} 

				     & $f(Y^{\ell+1})$         & $(H^{\ell+1}W^{\ell+1})\times D^{\ell+1}$                            		& $(p,q)$         	                      \\ \cline{2-4} 
				     & $f(X^{\ell+1})$         & $H^{\ell+1}\times W^{\ell+1}\times D^{\ell+1}$                       		& $(i^{\ell+1},j^{\ell+1},d^{\ell+1})$         	     \\ \hline

Output of Classical Pooling           & $\tilde{X}^{\ell+1}$ & $\frac{H^{\ell+1}}{P} \times \frac{W^{\ell+1}}{P} \times D^{\ell+1}$ &   $(\tilde{i}^{\ell+1},\tilde{j}^{\ell+1},\tilde{d}^{\ell+1})$     \\ \hline

\end{tabular}
\caption{Summary of variables describing outputs of the layer $\ell$, with the \emph{classical} algorithm.} 
\label{variable_summary_3}
\end{table}

Classical and quantum algorithms can be compared with these two diagrams:

\begin{equation}
\begin{cases}

\text{Quantum convolution layer : }
X^{\ell} \rightarrow 
\ket{\overline{X}^{\ell+1}} \rightarrow 
\ket{f(\overline{X}^{\ell+1})} \rightarrow 
\mathcal{X}^{\ell+1} \rightarrow 
\tilde{\mathcal{X}}^{\ell+1}\\

\text{Classical convolution layer : }
X^{\ell} \rightarrow 
X^{\ell+1} \rightarrow 
f(X^{\ell+1}) \rightarrow 
\tilde{X}^{\ell+1}

\end{cases}
\end{equation} 

~\\
We finally provide some remarks that could clarify some notations ambiguity:

- Formally, the output of the quantum algorithm is $\tilde{\mathcal{X}}^{\ell+1}$. It is used as input for the next layer $\ell+1$. But we consider that all variables' names are \emph{reset} when starting a new layer: $X^{\ell+1} \leftarrow \tilde{\mathcal{X}}^{\ell+1}$.

- For simplicity, we have sometimes replaced the indices $(i^{\ell+1},j^{\ell+1},d^{\ell+1})$ by $n$ to index the elements of the output.

- In Section \ref{pooling}, the input for layer $\ell+1$ is stored as $A^{\ell+1}$, for which the elements are indexed by $(p',r')$.

\section{Quantum Backprogation Algorithm}\label{backpropagation}

The entire QCNN is made of multiple layers. For the last layer's output, we expect only one possible outcome, or a few in the case of a classification task, which means that the dimension of the quantum output is very small. A full tomography can be performed on the last layer's output in order to calculate the outcome. The loss $\mathcal{L}$ is then calculated, as a measure of correctness of the predictions compared to the ground truth. As the classical CNN, our QCNN should be able to perform the optimization of its weights (elements of the kernels) to minimize the loss by an iterative method.

\begin{theorem}\label{theorem2}{ (Quantum Backpropagation for Quantum CNN) \\} 
Given the forward pass quantum algorithm in Algorithm \ref{QCNNLayer}, the input matrix $A^{\ell}$ and the kernel matrix $F^{\ell}$ stored in the QRAM for each layer $\ell$, and a loss function $\mathcal{L}$, there is a quantum backpropagation algorithm that 
estimates, for any precision $\delta> 0$, the gradient tensor $\frac{\partial \mathcal{L}}{\partial F^{\ell}}$ and update each element to perform gradient descent such that
\begin{equation}
\forall (s,q), \left|\frac{\partial \mathcal{L}}{\partial F^{\ell}_{s,q}} - \overline{\frac{\partial \mathcal{L}}{\partial F^{\ell}_{s,q}}} \right| 
\leq 2\delta \norm{\frac{\partial \mathcal{L}}{\partial F^{\ell}}}_{2}
\end{equation}
Let $\frac{\partial \mathcal{L}}{\partial Y^{\ell}}$ be the gradient with respect to the $\ell^{th}$ layer. The running time of a single layer $\ell$ for quantum backpropagation is given by
\begin{equation}
O\left(\left(\left(\mu(A^{\ell})+\mu(\frac{\partial \mathcal{L}}{\partial Y^{\ell+1}})\right)\kappa(\frac{\partial \mathcal{L}}{\partial F^{\ell}})+\left(\mu(\frac{\partial \mathcal{L}}{\partial Y^{\ell+1}})+\mu(F^{\ell})\right)\kappa(\frac{\partial \mathcal{L}}{\partial Y^{\ell}})\right) \frac{\log{1/\delta}}{\delta^{2}}\right)
\end{equation}
where for a matrix $V$, $\kappa(V)$ is the condition number and $\mu(V)$ is defined in Equation (\ref{mudefinition}).
\end{theorem}

Theorem \ref{theorem2} is proved in Section \ref{quantumalgobackprop} for the running time and Section \ref{gradientdescenterror} for the error guarantees.

\begin{algorithm} [H]
\caption{Quantum Backpropagation} \label{QBackpropagation}
\begin{algorithmic}[1]
\REQUIRE   Precision parameter $\delta$. Data matrices $A^{\ell}$ and kernel matrices $F^{\ell}$ stored in QRAM for each layer $\ell$.
\ENSURE Outputs gradient matrices $\frac{\partial \mathcal{L}}{\partial F^{\ell}}$ and $\frac{\partial \mathcal{L}}{\partial Y^{\ell}}$ for each layer $\ell$.\\
\vspace{10pt} 
\STATE Calculate the gradient for the last layer $L$ using the outputs and the true labels: $\frac{\partial \mathcal{L}}{\partial Y^{L}}$ 
\FOR{ $\ell = L-1, \cdots, 0$}

\STATE {\bf Step 1 : Modify the gradient}\\
With $\frac{\partial \mathcal{L}}{\partial Y^{\ell+1}}$ stored in QRAM, set to 0 some of its values to take into account pooling, tomography and non linearity that occurred in the forward pass of layer $\ell$. These values correspond to positions that haven't been sampled nor pooled, since they have no impact on the final loss.

\STATE {\bf Step 2 : Matrix-matrix multiplications}\\
With the modified values of $\frac{\partial \mathcal{L}}{\partial Y^{\ell+1}}$, use quantum linear algebra (Theorem \ref{QuantumLinearAlgebra}) to perform the following matrix-matrix multiplications
\begin{equation}
\begin{cases}
(A^{\ell})^{T} \cdot \frac{\partial L}{\partial Y^{\ell+1}}\\
\frac{\partial \mathcal{L}}{\partial Y^{\ell+1}} \cdot (F^{\ell})^T
\end{cases}
\end{equation}
to obtain quantum states corresponding to $\frac{\partial \mathcal{L}}{\partial F^{\ell}}$ and $\frac{\partial \mathcal{L}}{\partial Y^{\ell}}$.

\STATE {\bf Step 3 : $\ell_{\infty}$ tomography}\\
Using the $\ell_{\infty}$ tomography procedure given in Algorithm \ref{alg:tom}, estimate each entry of $\frac{\partial \mathcal{L}}{\partial F^{\ell}}$ and $\frac{\partial \mathcal{L}}{\partial Y^{\ell}}$ with errors $\delta \norm{\frac{\partial \mathcal{L}}{\partial F^{\ell}}}$ and $\delta \norm{\frac{\partial \mathcal{L}}{\partial Y^{\ell}}}$ respectively. 
Store all elements of $\frac{\partial \mathcal{L}}{\partial F^{\ell}}$ in QRAM. 
\STATE {\bf Step 4 : Gradient descent }\\
Perform gradient descent using the estimates from step 3 to update the values of $F^{\ell}$ in QRAM:
\begin{equation}
F^{\ell}_{s,q} \gets F^{\ell}_{s,q} - 
\lambda \left(\frac{\partial \mathcal{L}}{\partial F^{\ell}_{s,q}} 
\pm 2\delta\norm{\frac{\partial \mathcal{L}}{\partial F^{\ell}}}_{2} \right)
\end{equation}
\ENDFOR
\end{algorithmic}
\end{algorithm}

\subsection{Classical Backpropagation}
After each forward pass, the outcome is compared to the true labels and a suitable loss function is computed. We can update our weights by gradient descent to minimize this loss, and iterate. The main idea behind the backpropagation is to compute the derivatives of the loss $\mathcal{L}$, layer by layer, starting from the last one.

At layer $\ell$, the derivatives needed to perform the gradient descent are $\frac{\partial \mathcal{L}}{\partial F^{\ell}}$ and $\frac{\partial \mathcal{L}}{\partial Y^{\ell}}$. The first one represents the gradient of the final loss $\mathcal{L}$ with respect to each kernel element, a matrix of values that we will use to update the kernel weights $F^{\ell}_{s,q}$. The second one is the gradient of $\mathcal{L}$ with respect to the layer itself and is only needed to calculate the gradient $\frac{\partial \mathcal{L}}{\partial F^{\ell-1}}$ at layer $\ell-1$.

\subsubsection{Convolution Product}
We first consider a classical convolution layer without non linearity or pooling. Thus the output of layer $\ell$ is the same tensor as the input of layer $\ell+1$, namely $X^{\ell+1}$ or equivalently $Y^{\ell+1}$. Assuming we know $\frac{\partial \mathcal{L}}{\partial X^{\ell+1}}$ or equivalently $\frac{\partial \mathcal{L}}{\partial Y^{\ell+1}}$, both corresponding to the derivatives of the $(\ell + 1)^{th}$ layer's input, we will show how to calculate $\frac{\partial \mathcal{L}}{\partial F^{\ell}}$, the matrix of derivatives with respect to the elements of the previous kernel matrix $F^{\ell}$. This is the main goal in order to optimize the kernel's weights. 

The details of the following calculations can be found in \cite{CNNIntro}. We will use the notation $vec(X)$ to represents the vectorized form of any tensor $X$.

Recall that $A^{\ell}$ is the matrix expansion of the tensor $X^{\ell}$, whereas $Y^{\ell}$ is a matrix reshaping of $X^{\ell}$. By applying the chain rule $\frac{\partial \mathcal{L}}{\partial vec(F^{\ell})^T} = \frac{\partial \mathcal{L}}{\partial vec(X^{\ell+1})^T}\frac{\partial vec(X^{\ell+1})}{\partial vec(F^{\ell})^T}$, we can obtain (See \cite{CNNIntro} for calculations details):
\begin{equation}\label{updateF}
\frac{\partial \mathcal{L}}{\partial F^{\ell}} = (A^{\ell})^{T} \frac{\partial L}{\partial Y^{\ell+1}}
\end{equation}
Equation (\ref{updateF}) shows that, to obtain the desired gradient, we can just perform a matrix-matrix multiplication between the transposed layer itself ($A^{\ell}$) and the gradient with respect to the previous layer ($\frac{\partial L}{\partial Y^{\ell+1}}$).

Equation (\ref{updateF}) explains also why we will need to calculate $\frac{\partial \mathcal{L}}{\partial Y^{\ell}}$ in order to backpropagate through layer $\ell-1$. To calculate it, we use the chain rule again for $\frac{\partial \mathcal{L}}{\partial vec(X^{\ell})^T} = \frac{\partial \mathcal{L}}{\partial vec(X^{\ell+1})^T}\frac{\partial vec(X^{\ell+1})}{\partial vec(X^{\ell})^T}$. Recall that a point in $A^{\ell}$, indexed by the pair $(p,r)$, can correspond to several triplets $(i^{\ell},j^{\ell},d^{\ell})$ in $X^{\ell}$. We will use the notation $(p,r) \leftrightarrow ({i^{\ell},j^{\ell},d^{\ell}})$ to express formally this relation. One can show that $\frac{\partial \mathcal{L}}{\partial Y^{\ell+1}}(F^{\ell})^T$ is a matrix of same shape as $A^{\ell}$, and that the chain rule leads to a simple relation to calculate $\frac{\partial \mathcal{L}}{\partial Y^{\ell}}$ (See \cite{CNNIntro} for calculations details):

\begin{equation}\label{updateX}
\left[\frac{\partial \mathcal{L}}{\partial X^{\ell}}\right]_{i^{\ell},j^{\ell},d^{\ell}} = \sum_{(p,r) \leftrightarrow ({i^{\ell},j^{\ell},d^{\ell}})}\left[\frac{\partial \mathcal{L}}{\partial Y^{\ell+1}}(F^{\ell})^T\right]_{p,r}
\end{equation}

We have shown how to obtain the gradients with respect to the kernels $F^{\ell}$ and to the layer itself $Y^{\ell}$ (or equivalently $X^{\ell}$).

\subsubsection{Non Linearity}

The activation function has also an impact on the gradient. In the case of the ReLu, we should only cancel gradient for points with negative values. For points with positive value,  the derivatives remain the same since the function is the identity. A formal relation can be given by

\begin{equation}\label{relubackpropagation}
\left[\frac{\partial \mathcal{L}}{\partial X^{\ell+1}}\right]_{i^{\ell+1},j^{\ell+1},d^{\ell+1}} = 
\begin{cases}
    \left[\frac{\partial \mathcal{L}}{\partial f(X^{\ell+1})}\right]_{i^{\ell+1},j^{\ell+1},d^{\ell+1}} \text{ if } X^{\ell+1}_{i^{\ell+1},j^{\ell+1},d^{\ell+1}} \geq 0\\
    0 \text{ otherwise}\\
\end{cases}
\end{equation}

\subsubsection{Pooling}

If we take into account the pooling operation, we must change some of the gradients. Indeed, a pixel that hasn't been selected during pooling has no impact on the final loss, thus should have a gradient equal to 0.  We will focus on the case of Max Pooling (Average Pooling relies on similar idea). To state a formal relation, we will use the notations of Section \ref{pooling}: an element in the output of the layer, the tensor $f(X^{\ell+1})$, is located by the triplet $(i^{\ell+1},j^{\ell+1},d^{\ell+1})$. The tensor after pooling is noted $\tilde{X}^{\ell+1}$ and its points are located by the triplet $(\tilde{i}^{\ell+1},\tilde{j}^{\ell+1},\tilde{d}^{\ell+1})$. During backpropagation, after the calculation of $\frac{\partial \mathcal{L}}{\partial \tilde{X}^{\ell+1}}$, some of the derivatives of $f(X^{\ell+1})$ should be set to zero with the following rule:

\begin{equation}\label{poolingbackpropagation}
\left[\frac{\partial \mathcal{L}}{\partial f(X^{\ell+1})}\right]_{i^{\ell+1},j^{\ell+1},d^{\ell+1}} = 
\begin{cases}
    \left[\frac{\partial \mathcal{L}}{\partial \tilde{X}^{\ell+1}}\right]_{\tilde{i}^{\ell+1},\tilde{j}^{\ell+1},\tilde{d}^{\ell+1}} \text{ if $(i^{\ell+1},j^{\ell+1},d^{\ell+1})$ was selected during pooling}\\
    0 \text{ otherwise}\\
\end{cases}
\end{equation}

\subsection{Quantum Algorithm for Backpropagation}\label{quantumalgobackprop}

In this section, we want to give a quantum algorithm to perform backpropagation on a layer $\ell$, and detail the impact on the derivatives, given by the following diagram:

\begin{equation}\label{backpropagationdiagram}
\begin{cases}
    \frac{\partial \mathcal{L}}{\partial X^{\ell}}\\
    \frac{\partial \mathcal{L}}{\partial F^{\ell}}\\
\end{cases}
 \leftarrow 
\frac{\partial \mathcal{L}}{\partial \overline{X}^{\ell+1}}  \leftarrow 
\frac{\partial \mathcal{L}}{\partial f(\overline{X}^{\ell+1})}  \leftarrow 
\frac{\partial \mathcal{L}}{\partial \mathcal{X}^{\ell+1}}  \leftarrow 
\frac{\partial \mathcal{L}}{\partial \tilde{\mathcal{X}}^{\ell+1}}  =
\frac{\partial \mathcal{L}}{\partial X^{\ell+1}} 
\end{equation}

We assume that backpropagation has been done on layer $\ell+1$. 
This means in particular that $\frac{\partial \mathcal{L}}{\partial X^{\ell+1}}$ is stored in QRAM. 
However, as shown on Diagram (\ref{backpropagationdiagram}), $\frac{\partial \mathcal{L}}{\partial X^{\ell+1}}$ corresponds formally to $\frac{\partial \mathcal{L}}{\partial \tilde{\mathcal{X}}^{\ell+1}}$, and not $\frac{\partial \mathcal{L}}{\partial \overline{X}^{\ell+1}}$. 
Therefore, we will have to modify the values stored in QRAM to take into account non linearity, tomography and pooling. 
We will first consider how to implement $\frac{\partial \mathcal{L}}{\partial X^{\ell}}$ and $\frac{\partial \mathcal{L}}{\partial F^{\ell}}$ through backpropagation, considering only convolution product, as if $\frac{\partial \mathcal{L}}{\partial \overline{X}^{\ell+1}}$ and $\frac{\partial \mathcal{L}}{\partial X^{\ell+1}}$ where the same. Then we will detail how to simply modify $\frac{\partial \mathcal{L}}{\partial X^{\ell+1}}$ \emph{a priori}, by setting some of its values to 0.

\subsubsection{Quantum Convolution Product}
In this section we consider only the quantum convolution product without non linearity, tomography nor pooling, hence writing its output directly as $X^{\ell+1}$. Regarding derivatives, the quantum convolution product is equivalent to the classical one. Gradient relations (\ref{updateF}) and (\ref{updateX}) remain the same. Note that the $\epsilon$-approximation from Section \ref{registerencoding} doesn't participate in gradient considerations.

The gradient relations being the same, we still have to specify the quantum algorithm that implements the backpropagation and outputs classical description of $ \frac{\partial \mathcal{L}}{\partial X^{\ell}}$ and $\frac{\partial \mathcal{L}}{\partial F^{\ell}}$.
 We have seen that the two main calculations (\ref{updateF}) and (\ref{updateX}) are in fact matrix-matrix multiplications both involving $\frac{\partial \mathcal{L}}{\partial Y^{\ell+1}}$, the reshaped form of $\frac{\partial \mathcal{L}}{\partial X^{\ell+1}}$. For each, the classical running time is $O(H^{\ell+1}W^{\ell+1}D^{\ell+1}HWD^{\ell})$. 
We know from Theorem \ref{QuantumLinearAlgebra} and Theorem \ref{thm:tom} a quantum algorithm to perform efficiently a matrix-vector multiplication and return a classical state with $\ell_{\infty}$ norm guarantees. For a matrix $V$ and a vector $b$, both accessible from the QRAM, the running time to perform this operation is 
\begin{equation}
O\left(\frac{\mu(V) \kappa(V) \log{1/\delta}}{\delta^{2}}\right)
\end{equation}
where $\kappa(V)$ is the condition number of the matrix and $\mu(V)$ is a matrix parameter defined in Equation (\ref{mudefinition}). Precision parameter $\delta > 0$ is the error committed in the approximation for both Theorems \ref{QuantumLinearAlgebra} and \ref{thm:tom}.

We can therefore apply theses theorems to perform matrix-matrix multiplications, by simply decomposing them in several matrix-vector multiplications. For instance, in Equation (\ref{updateF}), the matrix could be $ (A^{\ell})^{T}$ and the different vectors would be each column of $\frac{\partial L}{\partial Y^{\ell+1}}$. The global running \cite{CGJ18} time to perform quantumly Equation (\ref{updateF}) is obtained by replacing $\mu(V)$ by $\mu(\frac{\partial \mathcal{L}}{\partial Y^{\ell+1}})+\mu(A^{\ell})$ and $\kappa(V)$ by $\kappa((A^{\ell})^T\cdot \frac{\partial \mathcal{L}}{\partial Y^{\ell+1}})$. Likewise, for Equation (\ref{updateX}), we have $\mu(\frac{\partial \mathcal{L}}{\partial Y^{\ell+1}})+\mu(F^{\ell})$ and $\kappa(\frac{\partial \mathcal{L}}{\partial Y^{\ell+1}}\cdot (F^{\ell})^T)$.

Note that the dimension of the matrix doesn't appear in the running time since we tolerate a $\ell_{\infty}$ norm guarantee for the error, instead of a $\ell_{2}$ guarantee (see Section \ref{proof:tom} for details). The reason why $\ell_\infty$ tomography is the right approximation here is because the result of these linear algebra operations are rows of the gradient matrices, that are not vectors in an euclidean space, but a series of numbers for which we want to be $\delta$-close to the exact values. See next section for more details. 

It is a open question to see if one can apply the same sub-sampling technique as in the forward pass (Section \ref{singlequantumlayer}) and sample only the highest derivatives of $\frac{\partial \mathcal{L}}{\partial X^{\ell}}$, to reduce the computation cost while maintaining a good optimization.

We then have to understand which elements of $\frac{\partial \mathcal{L}}{\partial X^{\ell+1}}$ must be set to zero to take into account the effects the non linearity, tomography and pooling.

\subsubsection{Quantum Non Linearity and Tomography}
To include the impact of the non linearity, one could apply the same rule as in (\ref{relubackpropagation}), and simply replace ReLu by capReLu. After the non linearity, we obtain $f(\overline{X}^{\ell+1})$, and the gradient relation would be given by 

\begin{equation}\label{fakecaprelubackpropagation}
\left[\frac{\partial \mathcal{L}}{\partial \overline{X}^{\ell+1}}\right]_{i^{\ell+1},j^{\ell+1},d^{\ell+1}} = 
\begin{cases}
    \left[\frac{\partial \mathcal{L}}{\partial f(\overline{X}^{\ell+1})}\right]_{i^{\ell+1},j^{\ell+1},d^{\ell+1}} \text{ if }  0\leq \overline{X}^{\ell+1}_{i^{\ell+1},j^{\ell+1},d^{\ell+1}} \leq C\\
    0 \text{ otherwise}\\
\end{cases}
\end{equation}

If an element of $\overline{X}^{\ell+1}$ was negative or bigger than the cap $C$, its derivative should be zero during the backpropagation. However, this operation was performed in quantum superposition. In the quantum algorithm, one cannot record at which positions $(i^{\ell+1},j^{\ell+1},d^{\ell+1})$ the activation function was selective or not. The gradient relation (\ref{fakecaprelubackpropagation}) cannot be implemented \emph{a posteriori}.

We provide a partial solution to this problem, using the fact that quantum tomography must also be taken into account for some derivatives. Indeed, only the points $(i^{\ell+1},j^{\ell+1},d^{\ell+1})$ that have been sampled should have an impact on the gradient of the loss. Therefore we replace the previous relation by

\begin{equation}\label{approximationnonlinearitybackpropagation}
\left[\frac{\partial \mathcal{L}}{\partial \overline{X}^{\ell+1}}\right]_{i^{\ell+1},j^{\ell+1},d^{\ell+1}} = 
\begin{cases}
    \left[\frac{\partial \mathcal{L}}{\partial \mathcal{X}^{\ell+1}}\right]_{i^{\ell+1},j^{\ell+1},d^{\ell+1}} \text{ if $(i^{\ell+1},j^{\ell+1},d^{\ell+1})$ was sampled }\\
    0 \text{ otherwise}\\
\end{cases}
\end{equation}

Nonetheless, we can argue that this approximation will be tolerable:

In the first case where $\overline{X}^{\ell+1}_{i^{\ell+1},j^{\ell+1},d^{\ell+1}} < 0$, the derivatives can not be set to zero as they should. But in practice, their values will be zero after the activation function and such points would not have a chance to be sampled. In conclusion their derivatives would be zero as required. 

In the other case where $\overline{X}^{\ell+1}_{i^{\ell+1},j^{\ell+1},d^{\ell+1}} > C$, the derivatives can not be set to zero as well but the points have a high probability of being sampled. Therefore their derivative will remain unchanged, as if we were using a ReLu instead of a capReLu. However in cases where the cap $C$ is high enough, this shouldn't be a source of disadvantage in practice.

\subsubsection{Quantum Pooling}
From relation (\ref{approximationnonlinearitybackpropagation}), we can take into account the impact of quantum pooling (see Section \ref{pooling}) on the derivatives. This case is easier since one can record the selected positions during the QRAM update. Therefore, applying the backpropagation is similar to the classical setting with Equation (\ref{poolingbackpropagation}).

\begin{equation}\label{quantumpoolingbackpropagation}
\left[\frac{\partial \mathcal{L}}{\partial \mathcal{X}^{\ell+1}}\right]_{i^{\ell+1},j^{\ell+1},d^{\ell+1}} = 
\begin{cases}
    \left[\frac{\partial \mathcal{L}}{\partial \tilde{\mathcal{X}}^{\ell+1}}\right]_{\tilde{i}^{\ell+1},\tilde{j}^{\ell+1},\tilde{d}^{\ell+1}} \text{ if $(i^{\ell+1},j^{\ell+1},d^{\ell+1})$ was selected during pooling}\\
    0 \text{ otherwise}\\
\end{cases}
\end{equation}

Note that we know $\frac{\partial \mathcal{L}}{\partial \tilde{\mathcal{X}}^{\ell+1}}$ as it is equal to $\frac{\partial \mathcal{L}}{\partial X^{\ell+1}}$, the gradient with respect to the input of layer $\ell+1$, known by assumption and stored in the QRAM.

\subsection{Conclusion and Running Time}

In conclusion, given $\frac{\partial \mathcal{L}}{\partial Y^{\ell+1}}$ in the QRAM, the quantum backpropagation first consists in applying the relations (\ref{quantumpoolingbackpropagation}) followed by (\ref{approximationnonlinearitybackpropagation}). The effective gradient now take into account non linearity, tomography and pooling that occurred during layer $\ell$. We can know use apply the quantum algorithm for matrix-matrix multiplication that implements relations (\ref{updateX}) and (\ref{updateF}).

Note that the steps in Algorithm \ref{QBackpropagation} could also be reversed: during backpropagation of layer $\ell+1$, when storing values for each elements of $\frac{\partial \mathcal{L}}{\partial Y^{\ell+1}}$ in the QRAM, one can already take into account (\ref{quantumpoolingbackpropagation}) and (\ref{approximationnonlinearitybackpropagation}) of layer $\ell$. In this case we directly store $\frac{\partial \mathcal{L}}{\partial \overline{X}^{\ell+1}}$, at no supplementary cost.

Therefore, the running time of the quantum backpropagation for one layer $\ell$, given as Algorithm \ref{QBackpropagation}, corresponds to the sum of the running times of the circuits for implementing relations (\ref{updateF}) and (\ref{updateX}). We finally obtain
\begin{equation}
O\left(\left(\left(\mu(A^{\ell})+\mu(\frac{\partial \mathcal{L}}{\partial Y^{\ell+1}})\right)\kappa((A^{\ell})^T \cdot \frac{\partial \mathcal{L}}{\partial Y^{\ell+1}})+\left(\mu(\frac{\partial \mathcal{L}}{\partial Y^{\ell+1}})+\mu(F^{\ell})\right)\kappa(\frac{\partial \mathcal{L}}{\partial Y^{\ell+1}}\cdot (F^{\ell})^T)\right) \frac{\log{1/\delta}}{\delta^{2}}\right), 
\end{equation}
which can be rewritten as
\begin{equation}
O\left(\left(\left(\mu(A^{\ell})+\mu(\frac{\partial \mathcal{L}}{\partial Y^{\ell+1}})\right)\kappa(\frac{\partial \mathcal{L}}{\partial F^{\ell}})+\left(\mu(\frac{\partial \mathcal{L}}{\partial Y^{\ell+1}})+\mu(F^{\ell})\right)\kappa(\frac{\partial \mathcal{L}}{\partial Y^{\ell}})\right) \frac{\log{1/\delta}}{\delta^{2}}\right). 
\end{equation}
Besides storing $\frac{\partial \mathcal{L}}{\partial X^{\ell}}$, the main output is a classical description of $\frac{\partial \mathcal{L}}{\partial F^{\ell}}$, necessary to perform gradient descent of the parameters of $F^{\ell}$.

\subsection{Gradient Descent and Classical equivalence}\label{gradientdescenterror}
In this part we will see the impact of the quantum backpropagation compared to the classical case, which can be reduced to a simple noise addition during the gradient descent. Recall that gradient descent, in our case, would consist in applying the following update rule

\begin{equation}
F^{\ell} \leftarrow F^{\ell} -\lambda \frac{\partial \mathcal{L}}{\partial F^{\ell}}
\end{equation}
with the learning rate $\lambda$.

Let's note $x = \frac{\partial \mathcal{L}}{\partial F^{\ell}}$ and its elements $x_{s,q} = \frac{\partial \mathcal{L}}{\partial F^{\ell}_{s,q}}$. From the first result of Theorem \ref{QuantumLinearAlgebra} with error $\delta<0$, and the tomography procedure from Theorem \ref{thm:tom}, with same error $\delta$, we can obtain a classical description of $\frac{\overline{x}}{\norm{\overline{x}}_2}$ with $\ell_{\infty}$ norm guarantee, such that:
$$
\norm{\frac{\overline{x}}{\norm{\overline{x}}_2} - \frac{x}{\norm{x}_2}}_{\infty} \leq \delta
$$ 
in time $\widetilde{O}(\frac{\kappa(V)\mu(V)\log(\delta)}{\delta^2})$, where we note $V$ is the matrix stored in the QRAM that allows to obtain $x$, as explained in Section \ref{quantumalgobackprop}. The $\ell_{\infty}$ norm tomography is used so that the error $\delta$ is at most the same for each component
$$
\forall (s,q), \left | \frac{\overline{x_{s,q}}}{\norm{\overline{x}}_2} - \frac{x_{s,q}}{\norm{x}_2} \right | \leq \delta
$$ 
From the second result of the Theorem \ref{QuantumLinearAlgebra} we can also obtain an estimate $\norm{\overline{x}}_2$ of the norm, for the same error $\delta$, such that
$$
|\norm{\overline{x}}_2 - \norm{x}_2| \leq \delta \norm{x}_2
$$ 
in time $\widetilde{O}(\frac{\kappa(V)\mu(V)}{\delta}\log(\delta))$ (which does not affect the overall asymptotic running time). Using both results we can obtain an unnormalized state close to $x$ such that, by the triangular inequality 
$$
\norm{\overline{x}-x}_{\infty} = \norm{\frac{\overline{x}}{\norm{\overline{x}}_2}\norm{\overline{x}}_2 - \frac{x}{\norm{x}_2}\norm{x}_2}_{\infty}
$$
$$
\leq 
\norm{\frac{\overline{x}}{\norm{\overline{x}}_2}\norm{\overline{x}}_2 - \frac{\overline{x}}{\norm{\overline{x}}_2}\norm{x}_2}_{\infty} + 
\norm{\frac{\overline{x}}{\norm{\overline{x}}_2}\norm{x}_2 - \frac{x}{\norm{x}_2}\norm{x}_2}_{\infty}
$$
$$
\leq 
1 \cdot  |\norm{\overline{x}}_2 - \norm{x}_2| +
\norm{x}_2 \cdot \norm{\frac{\overline{x}}{\norm{\overline{x}}_2} - \frac{x}{\norm{x}_2}}_{\infty}
$$
$$
\leq \delta \norm{x}_2  + \norm{\overline{x}}_2 \delta
\leq 2\delta\norm{x}_2
$$
in time $\widetilde{O}(\frac{\kappa(V)\mu(V)\log(\delta)}{\delta^2})$. In conclusion, with $\ell_{\infty}$ norm guarantee, having also access to the norm of the result is costless.

Finally, the noisy gradient descent update rule, expressed as $F^{\ell}_{s,q} \gets F^{\ell}_{s,q} - \lambda \overline{\frac{\partial \mathcal{L}}{\partial F^{\ell}_{s,q}}}$ can written in the worst case with
\begin{equation}
\overline{\frac{\partial \mathcal{L}}{\partial F^{\ell}_{s,q}}} =  
\frac{\partial \mathcal{L}}{\partial F^{\ell}_{s,q}} \pm 2\delta\norm{\frac{\partial \mathcal{L}}{\partial F^{\ell}}}_{2}
\end{equation}

To summarize, using the quantum linear algebra from Theorem \ref{QuantumLinearAlgebra} with $\ell_{\infty}$ norm tomography from Theorem \ref{thm:tom}, both with error $\delta$, along with norm estimation with relative error $\delta$ too, we can obtain classically the unnormalized values $\overline{\frac{\partial \mathcal{L}}{\partial F^{\ell}}}$ such that $\norm{\overline{\frac{\partial \mathcal{L}}{\partial F^{\ell}}}-\frac{\partial \mathcal{L}}{\partial F^{\ell}}}_{\infty} \leq 2\delta\norm{\frac{\partial \mathcal{L}}{\partial F^{\ell}}}_{2}$ or equivalently
\begin{equation}
\forall (s,q), \left|\overline{\frac{\partial \mathcal{L}}{\partial F^{\ell}_{s,q}}}-\frac{\partial \mathcal{L}}{\partial F^{\ell}_{s,q}}\right|  \leq  2\delta\norm{\frac{\partial \mathcal{L}}{\partial F^{\ell}}}_{2}
\end{equation}
Therefore the gradient descent update rule in the quantum case becomes $F^{\ell}_{s,q} \gets F^{\ell}_{s,q} - \lambda \overline{\frac{\partial \mathcal{L}}{\partial F^{\ell}_{s,q}}}$, which in the worst case becomes
\begin{equation}\label{gradientdescent}
F^{\ell}_{s,q} \gets F^{\ell}_{s,q} - 
\lambda \left(\frac{\partial \mathcal{L}}{\partial F^{\ell}_{s,q}} 
\pm 2\delta\norm{\frac{\partial \mathcal{L}}{\partial F^{\ell}}}_{2} \right)
\end{equation}

This proves the Theorem \ref{theorem2}. This update rule can be simulated by the addition of a random relative noise given as a gaussian centered on 0, with standard deviation equal to $\delta$. This is how we will simulate quantum backpropagation in the next Section.

Compared to the classical update rule, this corresponds to the addition of noise during the optimization step. This noise decreases as $\norm{\frac{\partial \mathcal{L}}{\partial F^{\ell}}}_{2}$, which is expected to happen while converging. Recall that the gradient descent is already a stochastic process. Therefore, we expect that such noise, with acceptable values of $\delta$, will not disturb the convergence of the gradient, as the following numerical simulations tend to confirm.

\section{Numerical Simulations}\label{NumericalSimulations}
As described above, the adaptation of the CNNs to the quantum setting implies some modifications that could alter the efficiency of the learning or classifying phases. We now present some experiments to show that such modified CNNs can converge correctly, as the original ones. 

The experiment, using the PyTorch library \cite{paszke2017automatic}, consists of training classically a small convolutional neural network for which we have added a ``quantum" sampling after each convolution, as in Section \ref{qramupdate}. Instead of parametrizing it with the precision $\eta$, we have choosed to use the sampling ratio $\sigma$ that represents the number of samples drawn during tomography. This two definitions are equivalent, as shown in Section \ref{tomographyconvolution}, but the second one is more intuitive regarding the running time and the simulations. 

We also add a noise simulating the amplitude estimation (Section \ref{registerencoding}, parameter: $\epsilon$), followed by a capReLu instead of the usual ReLu (Section \ref{capReLu}, parameter: $C$), and a noise during the backpropagation (Section \ref{gradientdescenterror}, parameter: $\delta$). In the following results, we observe that our quantum CNN is able to learn and classify visual data from the widely used MNIST dataset. This dataset is made of 60.000 training images and 10.000 testing images of handwritten digits. Each image is a 28x28 grayscale pixels between 0 and 255 (8 bits encoding), before normalization.

Let's first observe the ``quantum" effects on an image of the dataset. In particular, the effect of the capped non linearity, the introduction of noise and the quantum sampling.

\begin{figure}[H]
\minipage{0.25\textwidth}
  \includegraphics[width=\linewidth]{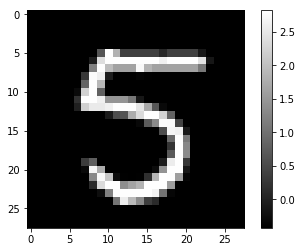}
  \includegraphics[width=\linewidth]{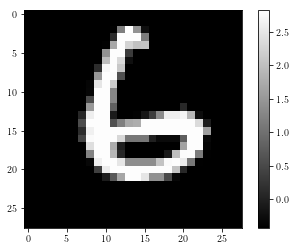}
\endminipage\hfill
\minipage{0.25\textwidth}
  \includegraphics[width=\linewidth]{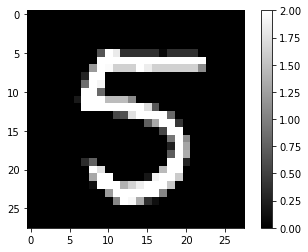}
  \includegraphics[width=\linewidth]{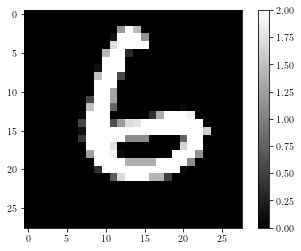}
\endminipage\hfill
\minipage{0.25\textwidth}
  \includegraphics[width=\linewidth]{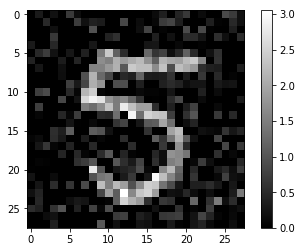}
  \includegraphics[width=\linewidth]{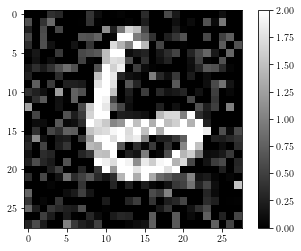}
\endminipage
\minipage{0.25\textwidth}
  \includegraphics[width=\linewidth]{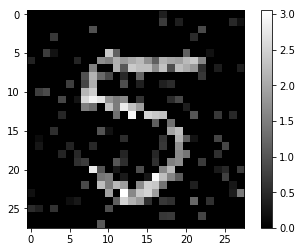}
  \includegraphics[width=\linewidth]{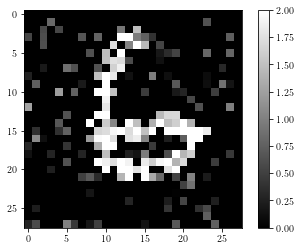}
\endminipage
\caption{Effects of the QCNN on a 28x28 input image. From left to right: original image, image after applying a capReLu activation function with a cap $C$ at $2.0$, introduction of a strong noise during amplitude estimation with $\epsilon=0.5$, quantum sampling with ratio $\sigma=0.4$ that samples the highest values in priority. The useful information tends to be conserved in this example. The side gray scale indicates the value of each pixel. Note that during the QCNN layer, a convolution is supposed to happen before the last image but we chose not to perform it for better visualization.}
\label{fig:quantumeffects}
\end{figure}

We now present the full simulation of our quantum CNN. In the following, we use a simple network made of 2 convolution layers, and compare our quantum CNN to the classical one. The first and second layers are respectively made of 5 and 10 kernels, both of size 7x7. A three-layer fully connected network is applied at the end and a softmax activation function is applied on the last layer to detect the predicted outcome over 10 classes (the ten possible digits). Note that we didn't introduce pooling, being equivalent between quantum and classical algorithms and not improving the results on our CNN. The objective of the learning phase is to minimize the loss function, defined by the negative log likelihood of the classification on the training set. The optimizer used was a built-in Stochastic Gradient Descent. 

Using PyTorch, we have been able to implement the following quantum effects (the first three points are shown in Figure \ref{fig:quantumeffects}):
\begin{itemize}
\item The addition of a noise, to simulate the approximation of amplitude estimation during the forward quantum convolution layer, by adding gaussian noise centered on 0 and with standard deviation $2M\epsilon$, with $M = \max_{p,q}{\norm{A_{p}}\norm{F_{q}}}$, as given by Equation (\ref{errorAEfinal}). 
\item A modification of the non linearity: a ReLu function that becomes constant above the value $T$ (the cap).
\item A sampling procedure to apply on a tensor with a probability distribution proportional to the tensor itself, reproducing the quantum sampling with ratio $\sigma$.
\item The addition of a noise during the gradient descent, to simulate the quantum backpropagation, by adding a gaussian noise centered on 0 with standard deviation $\delta$, multiplied by the norm of the gradient, as given by Equation (\ref{gradientdescent}).\\
\end{itemize}

The CNN used for this simulation may seem ``small" compared to the standards AlexNet \cite{krizhevsky2012imagenet} or VGG-16 \cite{simonyan2014very}, or those used in industry. However simulating this small QCNN on a classical computer was already very computationally intensive and time consuming, due to the``quantum" sampling task, apparently not optimized for a classical implementation in PyTorch. Every single training curve showed in Figure \ref{experiment_1_2} could last for 4 to 8 hours. Hence adding more convolutional layers wasn't convenient. Similarly, we didn't compute the loss on the whole testing set (10.000 images) during the the training to plot the testing curve. However we have computed the test losses and accuracies once the model trained (see Table \ref{testset_metrics_1}), in order to detect potential overfitting cases.

We now present the result of the training phase for a quantum version of this CNN, where partial quantum sampling is applied, for different sampling ratio (number of samples taken from the resulting convolution). Since the quantum sampling gives more probability to observe high value pixels, we expect to be able to learn correctly even with small ratio ($\sigma \leq 0.5$). We compare these training curve to the classical one. The learning has been done on two epochs, meaning that the whole dataset is used twice. The following plots show the evolution of the loss $\mathcal{L}$ during the iterations on batches. This is the standard indicator of the good convergence of a neural network learning phase. We can compare the evolution of the loss between a classical CNN and our QCNN for different parameters.

\begin{figure}[H]
\centering
\includegraphics[width=120mm] {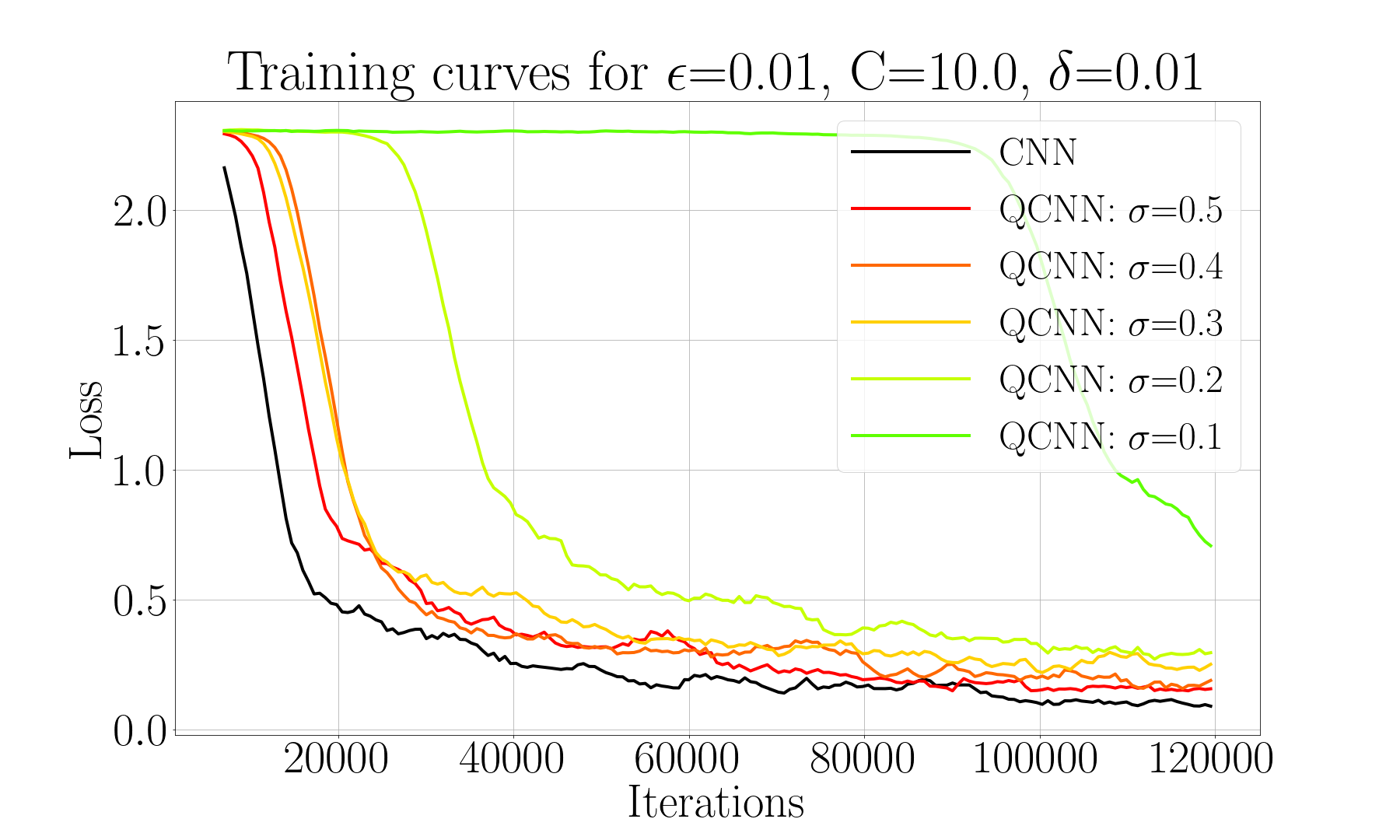} 
\captionsetup{justification=raggedright, margin=1cm}
\caption{Training curves comparison between the classical CNN and the Quantum CNN (QCNN) for $\epsilon=0.01$, $C=10$, $\delta=0.01$ and the sampling ratio $\sigma$ from $0.1$ to $0.5$. We can observe a learning phase similar to the classical one, even for a weak sampling of 20\% or 30\% of each convolution output, which tends to show that the meaningful information is distributed only at certain location of the images, coherently with the purpose of the convolution layer. Even for a very low sampling ratio of 10\%, we observe a convergence despite a late start.}
\label{Training curves }
\end{figure}

\begin{figure}[H]
\begin{multicols}{2}
    \includegraphics[width=\linewidth]{training_epsilon001_T10_delta001.png}\par
    \includegraphics[width=\linewidth]{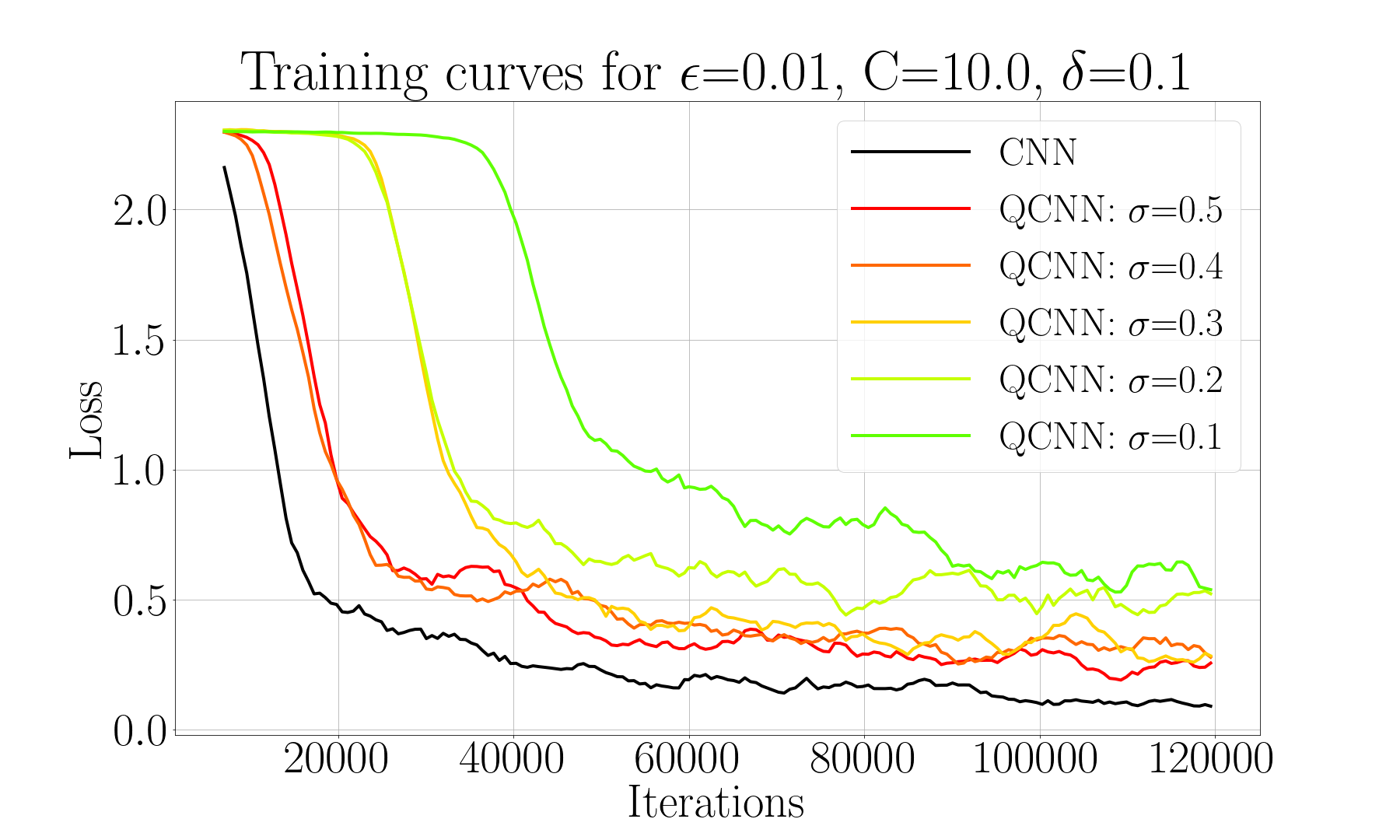}\par 
\end{multicols}
\begin{multicols}{2}
    \includegraphics[width=\linewidth]{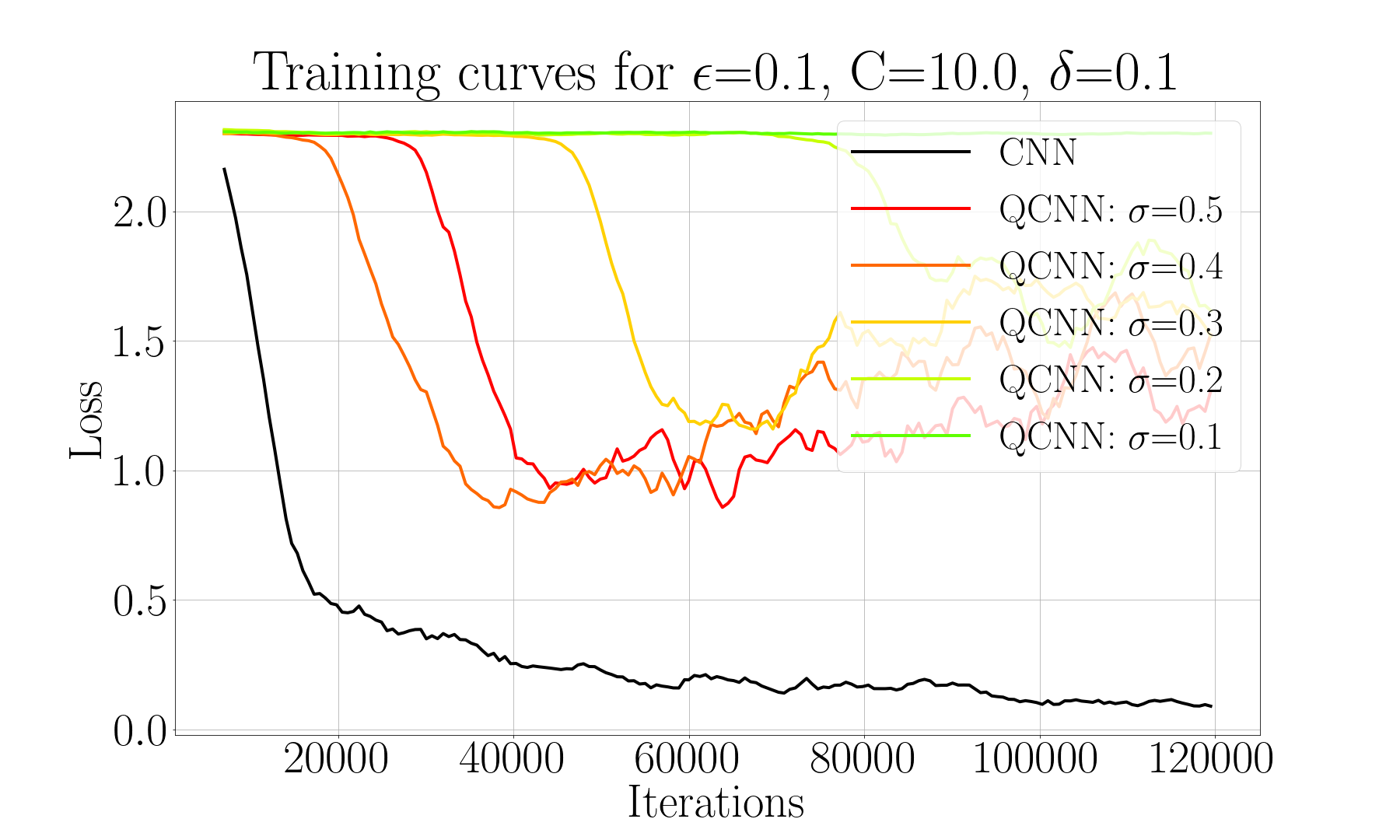}\par
    \includegraphics[width=\linewidth]{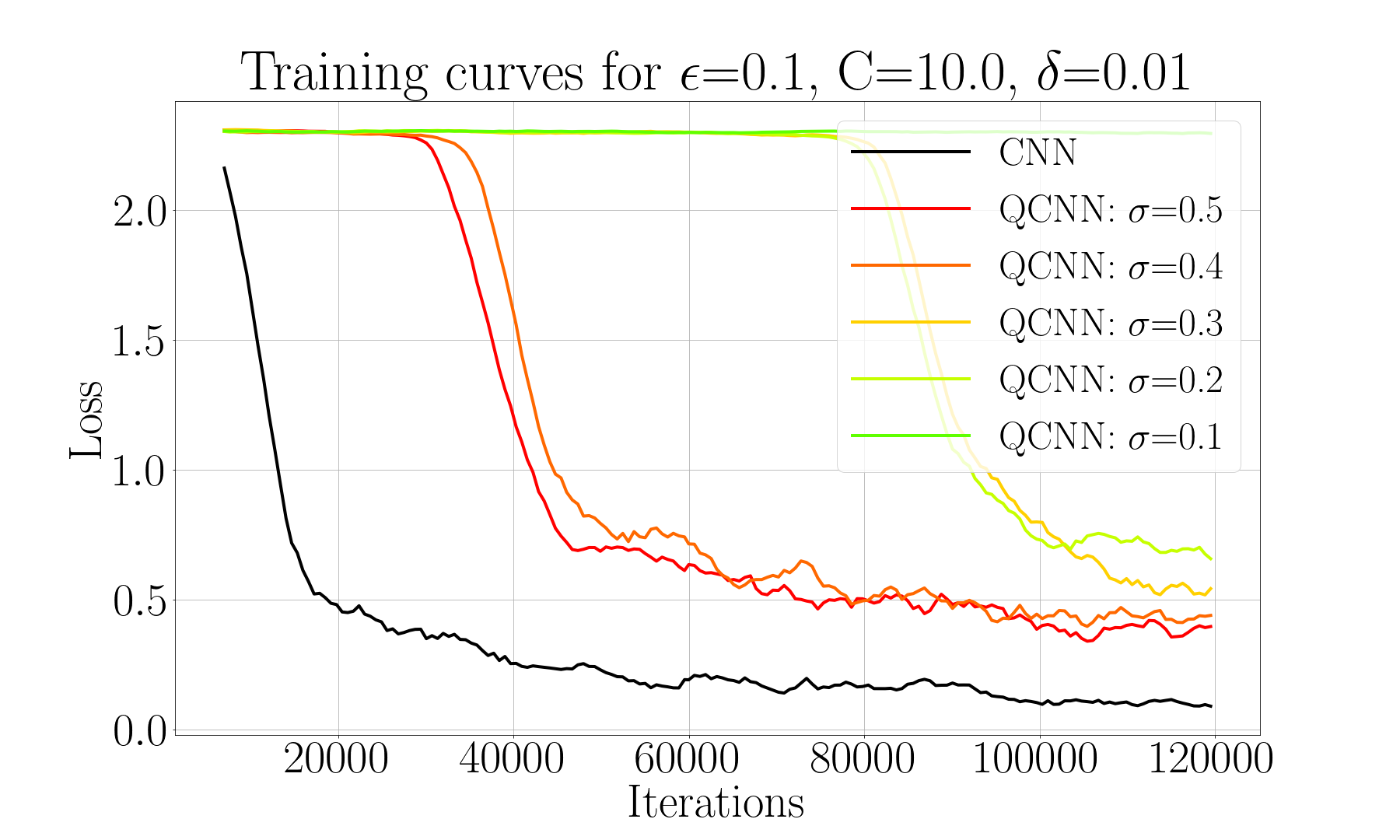}\par 
\end{multicols}
\begin{multicols}{2}
    \includegraphics[width=\linewidth]{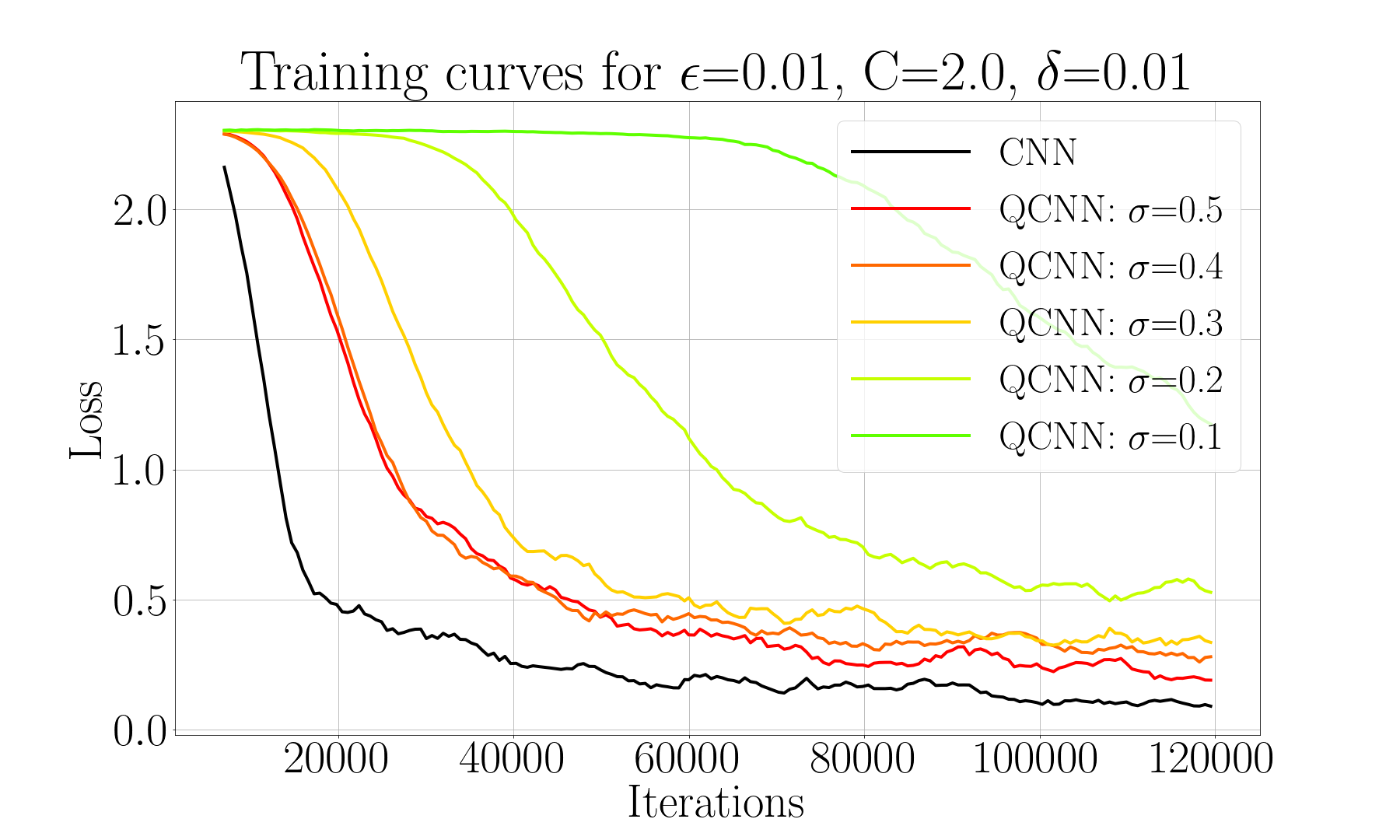}\par 
    \includegraphics[width=\linewidth]{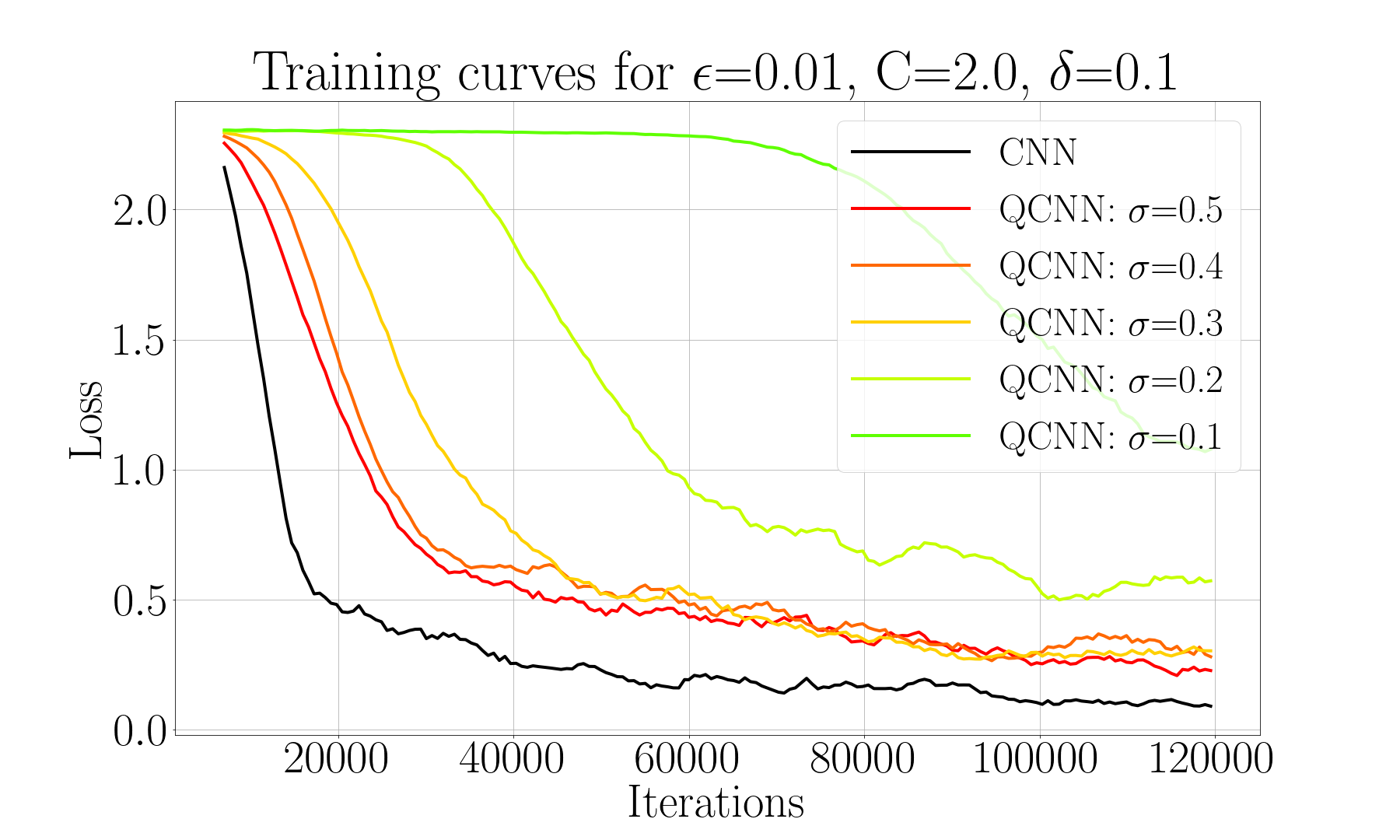}\par 
    \end{multicols}
\caption{Numerical simulations of the training of the QCNN. These training curves represent the evolution of the Loss $\mathcal{L}$ as we iterate through the MNIST dataset. For each graph, the amplitude estimation error $\epsilon$ $(0.1,0.01)$, the non linearity cap $C$ $(2,10)$, and the backpropagation error $\delta$ $(0.1,0.01)$ are fixed whereas the quantum sampling ratio $\sigma$ varies from 0.1 to 0.5. We can compare each training curve to the classical learning (CNN). Note that these training curves are smoothed, over windows of 12 steps, for readability.}
\label{experiment_1_2}
\end{figure}

In the following we report the classification results of the QCNN when applied on the test set (10.000 images). We distinguish to use cases: in Table \ref{testset_metrics_1} the QCNN has been trained quantumly as described in this paper, whereas in Table \ref{testset_metrics_2} we first have trained the classical CNN, then transferred the weights to the QCNN only for the classification. This second use case has a global running time worst than the first one, but we see it as another concrete application: quantum machine learning could be used only for faster classification from a classically generated model, which could be the case for high rate classification task (e.g. for autonomous systems, classification over many simultaneous inputs). We report the test loss and accuracy for different values of the sampling ratio $\sigma$, the amplitude estimation error $\epsilon$, and for the backpropagation noise $\delta$ in the first case. The cap $C$ is fixed at 10. These values must be compared to the classical CNN classification metrics, for which the loss is 0.129 and the accuracy is 96.1\%. Note that we used a relatively small CNN and hence the accuracy is just over $96\%$, lower than the best possible accuracy with larger CNN. \\

\begin{table}[h]
\centering
\begin{tabular}{|c|c|c|c|c|c|}
\hline
\multicolumn{6}{|c|}{QCNN Test - Classification}                                                          \\ \hline
\multirow{2}{*}{$\sigma$} & $\epsilon$ & \multicolumn{2}{c|}{0.01} & \multicolumn{2}{c|}{0.1} \\ \cline{2-6} 
                                  & $\delta$   & 0.01         & 0.1        & 0.01        & 0.1        \\ \hline
\multirow{2}{*}{0.1}              & Loss       & 0.519            & 0.773          & 2.30           & 2.30          \\
                                  & Accuracy   & 82.8\%          & 74.8\%        & 11.5\%         & 11.7\%        \\ \hline
\multirow{2}{*}{0.2}              & Loss       & 0.334            & 0.348          & 0.439           & 1.367          \\
                                  & Accuracy   & 89.5\%          & 89.0\%        & 86.2\%         & 54.1\%        \\ \hline
\multirow{2}{*}{0.3}              & Loss       & 0.213            & 0.314          & 0.381           & 0.762          \\
                                  & Accuracy   & 93.4\%          & 90.3\%        & 87.9\%         & 76.8\%        \\ \hline
\multirow{2}{*}{0.4}              & Loss       & 0.177            & 0.215          & 0.263           & 1.798          \\
                                  & Accuracy   & 94.7\%          & 93.3\%        & 91.8\%         & 34.9\%        \\ \hline
\multirow{2}{*}{0.5}              & Loss       & 0.142            & 0.211          & 0.337           & 1.457          \\
                                  & Accuracy   & 95.4\%          & 93.5\%        & 89.2\%         & 52.8\%        \\ \hline
\end{tabular}
\caption{QCNN trained with quantum backpropagation on MNIST dataset. With $C=10$ fixed.}
\label{testset_metrics_1}
\end{table}

\begin{table}[h]
\centering
\begin{tabular}{|c|c|c|c|}
\hline
\multicolumn{4}{|c|}{QCNN Test - Classification}   \\ \hline
$\sigma$             & $\epsilon$ & 0.01 & 0.1 \\ \hline
\multirow{2}{*}{0.1} & Loss       & 1.07    & 1.33   \\
                     & Accuracy   & 86.1\%  & 78.6\% \\ \hline
\multirow{2}{*}{0.2} & Loss       & 0.552    & 0.840   \\
                     & Accuracy   & 92.8\%  & 86.5\% \\ \hline
\multirow{2}{*}{0.3} & Loss       & 0.391    & 0.706   \\
                     & Accuracy   & 94,3\%  & 85.8\% \\ \hline
\multirow{2}{*}{0.4} & Loss       & 0.327    & 0.670   \\
                     & Accuracy   & 94.4\%  & 84.0\% \\ \hline
\multirow{2}{*}{0.5} & Loss       & 0.163    & 0.292   \\
                     & Accuracy   & 95.9\%  & 93.5\% \\ \hline
\end{tabular}
\caption{QCNN created from a classical CNN trained on MNIST dataset. With $\delta = 0.01$ and $C=10$ fixed.} 
\label{testset_metrics_2}
\end{table}

Our simulations show that the QCNN is able to learn despite the introduction of noise, tensor sampling and other modifications. In particular it shows that only a fraction of the information is meaningful for the neural network, and that the quantum algorithm captures this information in priority. This learning can be more or less efficient depending on the choice of the key parameters. For reasonable values of these parameters, the QCNN is able to converge during the training phase. It can then classify correctly on both training and testing set, indicating that it does not overfit the data.  

We notice that the learning curves sometimes present a late start before the convergence initializes, in particular for small sampling ratio. This late start can be due to the random initialization of the kernel weights, that performs a meaningless convolution, a case where the quantum sampling of the output is of no interest. However it is very interesting to see that despite this late start, the kernel can start converging once they have found a good combination.

Overall, it is possible that the QCNN presents some behaviors that do not have a classical equivalent. Understanding their potential effects, positive or negative, is an open question, all the more so as the effects of the classical CNN's hyperparameters are already a topic an active research \cite{samek2017explainable}. Note also that the size of the neural network used in this simulation is rather small. A following experiment would be to simulate a quantum version of a standard deeper CNN (AlexNet or VGG-16), eventually on more complex dataset, such as CIFAR-10 \cite{krizhevsky2009learning} or Fashion MNIST \cite{xiao2017fashion}.

\section{Conclusions}\label{conclusions}
We have presented a quantum algorithm for evaluating and training convolutional neural networks (CNN). At the core of this algorithm, we have developed the first quantum algorithm for computing a convolution product between two tensors, with a substantial speed up. This technique could be reused in other signal processing tasks that would benefit an enhancement by a quantum computer. Layer by layer, convolutional neural networks process and extract meaningful information. Following this idea of learning foremost important features, we have proposed a new approach of quantum tomography where the most meaningful information is sampled with higher probability, hence reducing the complexity of our algorithm. 

Our Quantum CNN is complete in the sense that almost all classical architectures can be implemented in a quantum fashion: any (non negative an upper bounded) non linearity, pooling, number of layers and size of kernels are available. Our circuit is shallow, indeed one could repeat the main loop many times on the same shallow circuit, since performing the convolution product uses shallow linear algebra techniques, and is similar for all layers. The pooling and non linearity are included in the loop. Our building block approach, layer by layer, allows a high modularity, and can be combined with previous works on quantum feedforward neural network \cite{kerenidis2018neural}. 

The running time presents a speedup compared to the classical algorithm, due to fast linear algebra when computing the convolution product, and by only sampling the important values from the resulting quantum state. This speedup can be highly significant in cases where the number of channels $D^{\ell}$ in the input tensor is high (high dimensional time series, videos sequences, games play) or when the number of kernels $D^{\ell+1}$ is big, allowing deep architectures for CNN, which was the case in the recent breakthrough of DeepMind AlphaGo algorithm \cite{silver2016mastering}. The Quantum CNN also allows larger kernels, that could be used for larger input images, since the size the kernels must be a contant fraction of the input in order to recognize patterns. However, despite our new techniques to reduce the complexity, applying a non linearity and reusing the result of a layer for the next layer make register encoding and state tomography mandatory, hence preventing from having an exponential speedup on the number of input parameters. 

Finally we have presented a backpropagation algorithm that can also be implemented as a quantum circuit. The numerical simulations on a small CNN show that despite the introduction of noise and sampling, the QCNN can efficiently learn to classify visual data from the MNIST dataset, performing a similar accuracy than the classical CNN. 

\paragraph{Acknowledgments} The work has been partially supported by projects ANR quBIC, ANR quData and QuantERA QuantAlgo.

\bibliographystyle{IEEEtran} 
\bibliography{Q-CNN} 

\begin{thebibliography}{10}
\providecommand{\url}[1]{#1}
\csname url@samestyle\endcsname
\providecommand{\newblock}{\relax}
\providecommand{\bibinfo}[2]{#2}
\providecommand{\BIBentrySTDinterwordspacing}{\spaceskip=0pt\relax}
\providecommand{\BIBentryALTinterwordstretchfactor}{4}
\providecommand{\BIBentryALTinterwordspacing}{\spaceskip=\fontdimen2\font plus
\BIBentryALTinterwordstretchfactor\fontdimen3\font minus
  \fontdimen4\font\relax}
\providecommand{\BIBforeignlanguage}[2]{{%
\expandafter\ifx\csname l@#1\endcsname\relax
\typeout{** WARNING: IEEEtran.bst: No hyphenation pattern has been}%
\typeout{** loaded for the language `#1'. Using the pattern for}%
\typeout{** the default language instead.}%
\else
\language=\csname l@#1\endcsname
\fi
#2}}
\providecommand{\BIBdecl}{\relax}
\BIBdecl

\bibitem{YLeCun}
Y.~LeCun, L.~Bottou, Y.~Bengio, and P.~Haffner, ``Gradient-based learning
  applied to document recognition,'' \emph{Proceedings of the IEEE}, 1998.

\bibitem{HandbookCNN}
Y.~LeCun and Y.~Bengio, \emph{The Handbook of Brain Theory and Neural
  Networks}.\hskip 1em plus 0.5em minus 0.4em\relax MIT Press, 1995.

\bibitem{schuld2014quest}
M.~Schuld, I.~Sinayskiy, and F.~Petruccione, ``The quest for a quantum neural
  network,'' \emph{Quantum Information Processing}, vol.~13, no.~11, pp.
  2567--2586, 2014.

\bibitem{qmeans}
I.~Kerenidis, J.~Landman, A.~Luongo, and A.~Prakash, ``q-means: A quantum
  algorithm for unsupervised machine learning,'' \emph{arXiv preprint
  arXiv:1812.03584}, 2018.

\bibitem{LMR13}
\BIBentryALTinterwordspacing
S.~Lloyd, M.~Mohseni, and P.~Rebentrost, ``{Quantum algorithms for supervised
  and unsupervised machine learning},'' \emph{arXiv}, vol. 1307.0411, pp.
  1--11, 7 2013. [Online]. Available: \url{http://arxiv.org/abs/1307.0411}
\BIBentrySTDinterwordspacing

\bibitem{lloyd2014quantum}
S.~Lloyd, M.~Mohseni, and P.~Rebentrost, ``Quantum principal component
  analysis,'' \emph{Nature Physics}, vol.~10, no.~9, p. 631, 2014.

\bibitem{KP17}
I.~Kerenidis and A.~Prakash, ``Quantum gradient descent for linear systems and
  least squares,'' \emph{arXiv:1704.04992}, 2017.

\bibitem{WKS14}
\BIBentryALTinterwordspacing
N.~Wiebe, A.~Kapoor, and K.~M. Svore, ``{Quantum Algorithms for
  Nearest-Neighbor Methods for Supervised and Unsupervised Learning},''
  \emph{arXiv:1401.2142v2}, 2014. [Online]. Available:
  \url{https://arxiv.org/pdf/1401.2142.pdf}
\BIBentrySTDinterwordspacing

\bibitem{kerenidis2018neural}
J.~Allcock, C.-Y. Hsieh, I.~Kerenidis, and S.~Zhang, ``Quantum algorithms for
  feedforward neural networks,'' \emph{arXiv preprint arXiv:1812.03089}, 2018.

\bibitem{rebentrost2018quantum}
P.~Rebentrost, T.~R. Bromley, C.~Weedbrook, and S.~Lloyd, ``Quantum hopfield
  neural network,'' \emph{Physical Review A}, vol.~98, no.~4, p. 042308, 2018.

\bibitem{wiebe2014quantum}
N.~Wiebe, A.~Kapoor, and K.~M. Svore, ``Quantum deep learning,'' \emph{arXiv
  preprint arXiv:1412.3489}, 2014.

\bibitem{krizhevsky2012imagenet}
A.~Krizhevsky, I.~Sutskever, and G.~E. Hinton, ``Imagenet classification with
  deep convolutional neural networks,'' in \emph{Advances in neural information
  processing systems}, 2012, pp. 1097--1105.

\bibitem{bojarski2016visualbackprop}
M.~Bojarski, A.~Choromanska, K.~Choromanski, B.~Firner, L.~Jackel, U.~Muller,
  and K.~Zieba, ``Visualbackprop: efficient visualization of cnns,''
  \emph{arXiv preprint arXiv:1611.05418}, 2016.

\bibitem{george2018deep}
D.~George and E.~Huerta, ``Deep learning for real-time gravitational wave
  detection and parameter estimation: Results with advanced ligo data,''
  \emph{Physics Letters B}, vol. 778, pp. 64--70, 2018.

\bibitem{goodfellow2014generative}
I.~Goodfellow, J.~Pouget-Abadie, M.~Mirza, B.~Xu, D.~Warde-Farley, S.~Ozair,
  A.~Courville, and Y.~Bengio, ``Generative adversarial nets,'' in
  \emph{Advances in neural information processing systems}, 2014, pp.
  2672--2680.

\bibitem{beer2019efficient}
K.~Beer, D.~Bondarenko, T.~Farrelly, T.~J. Osborne, R.~Salzmann, and R.~Wolf,
  ``Efficient learning for deep quantum neural networks,'' \emph{arXiv preprint
  arXiv:1902.10445}, 2019.

\bibitem{cong2018quantum}
I.~Cong, S.~Choi, and M.~D. Lukin, ``Quantum convolutional neural networks,''
  \emph{arXiv preprint arXiv:1810.03787}, 2018.

\bibitem{farhi2018classification}
E.~Farhi and H.~Neven, ``Classification with quantum neural networks on near
  term processors,'' \emph{arXiv preprint arXiv:1802.06002}, 2018.

\bibitem{henderson2019quanvolutional}
M.~Henderson, S.~Shakya, S.~Pradhan, and T.~Cook, ``Quanvolutional neural
  networks: Powering image recognition with quantum circuits,'' \emph{arXiv
  preprint arXiv:1904.04767}, 2019.

\bibitem{killoran2018continuous}
N.~Killoran, T.~R. Bromley, J.~M. Arrazola, M.~Schuld, N.~Quesada, and
  S.~Lloyd, ``Continuous-variable quantum neural networks,'' \emph{arXiv
  preprint arXiv:1806.06871}, 2018.

\bibitem{preskill2018quantum}
J.~Preskill, ``Quantum computing in the nisq era and beyond,'' \emph{Quantum},
  vol.~2, p.~79, 2018.

\bibitem{CNNIntro}
J.~Wu, ``Introduction to convolutional neural networks,''
  \emph{https://pdfs.semanticscholar.org/450c/a19932fcef1ca6d0442cbf52fec38fb9d1e5.pdf},
  2017.

\bibitem{KP16}
I.~Kerenidis and A.~Prakash, ``Quantum recommendation systems,''
  \emph{Proceedings of the 8th Innovations in Theoretical Computer Science
  Conference}, 2017.

\bibitem{BHMT00}
G.~Brassard, P.~Hoyer, M.~Mosca, and A.~Tapp, ``Quantum amplitude amplification
  and estimation,'' \emph{Contemporary Mathematics}, vol. 305, pp. 53--74,
  2002.

\bibitem{CGJ18}
S.~Chakraborty, A.~Gily{\'e}n, and S.~Jeffery, ``The power of block-encoded
  matrix powers: improved regression techniques via faster {Hamiltonian}
  simulation,'' \emph{arXiv preprint arXiv:1804.01973}, 2018.

\bibitem{KP18}
I.~Kerenidis and A.~Prakash, ``A quantum interior point method for {LPs} and
  {SDPs},'' \emph{arXiv:1808.09266}, 2018.

\bibitem{goodfellow2016deep}
I.~Goodfellow, Y.~Bengio, and A.~Courville, \emph{Deep learning}.\hskip 1em
  plus 0.5em minus 0.4em\relax MIT press, 2016.

\bibitem{bishop1995training}
C.~M. Bishop, ``Training with noise is equivalent to tikhonov regularization,''
  \emph{Neural computation}, vol.~7, no.~1, pp. 108--116, 1995.

\bibitem{paszke2017automatic}
A.~Paszke, S.~Gross, S.~Chintala, G.~Chanan, E.~Yang, Z.~DeVito, Z.~Lin,
  A.~Desmaison, L.~Antiga, and A.~Lerer, ``Automatic differentiation in
  pytorch,'' 2017.

\bibitem{simonyan2014very}
K.~Simonyan and A.~Zisserman, ``Very deep convolutional networks for
  large-scale image recognition,'' \emph{arXiv preprint arXiv:1409.1556}, 2014.

\bibitem{samek2017explainable}
W.~Samek, T.~Wiegand, and K.-R. M{\"u}ller, ``Explainable artificial
  intelligence: Understanding, visualizing and interpreting deep learning
  models,'' \emph{arXiv preprint arXiv:1708.08296}, 2017.

\bibitem{krizhevsky2009learning}
A.~Krizhevsky and G.~Hinton, ``Learning multiple layers of features from tiny
  images,'' Citeseer, Tech. Rep., 2009.

\bibitem{xiao2017fashion}
H.~Xiao, K.~Rasul, and R.~Vollgraf, ``Fashion-mnist: a novel image dataset for
  benchmarking machine learning algorithms,'' \emph{arXiv preprint
  arXiv:1708.07747}, 2017.

\bibitem{silver2016mastering}
D.~Silver, A.~Huang, C.~J. Maddison, A.~Guez, L.~Sifre, G.~Van Den~Driessche,
  J.~Schrittwieser, I.~Antonoglou, V.~Panneershelvam, M.~Lanctot \emph{et~al.},
  ``Mastering the game of go with deep neural networks and tree search,''
  \emph{nature}, vol. 529, no. 7587, p. 484, 2016.

\bibitem{tang2018quantum}
E.~Tang, ``Quantum-inspired classical algorithms for principal component
  analysis and supervised clustering,'' \emph{arXiv preprint arXiv:1811.00414},
  2018.

\bibitem{arrazola2019quantum}
J.~M. Arrazola, A.~Delgado, B.~R. Bardhan, and S.~Lloyd, ``Quantum-inspired
  algorithms in practice,'' \emph{arXiv preprint arXiv:1905.10415}, 2019.

\end{thebibliography}

\newpage

\section{Appendix : Algorithm and Proof for $\ell_{\infty}$ norm tomography}\label{proof:tom}
We prove the Theorem \ref{thm:tom} introduced in this paper, and present the algorithm that allows to obtain a tomography on a quantum state with $\ell_{\infty}$ norm guarantee. The algorithm and the proof closely follows the work of \cite{KP18} on the $\ell_2$-norm tomography, but requires exponentially less resources. In the following we consider a quantum state $\ket{x} = \sum_{i \in [d]} x_i \ket{i}$, with $x \in \R^d$ and $\norm{x}=1$.

\begin{algorithm} [h]
\caption{$\ell_{\infty}$ norm tomography} \label{alg:tom}
\begin{algorithmic}[1]
\REQUIRE   Error $\delta > 0$, access to unitary $U : \ket{0} \mapsto \ket{x} = \sum_{i \in [d]}x_i\ket{i}$, the controlled version of $U$, QRAM access. 
\ENSURE Classical vector $\widetilde{X} \in \R^{d}$, such that $\norm{\widetilde{X}}=1$ and $\norm{\widetilde{X}-x}_{\infty} < \delta$.
\vspace{10pt} 
\STATE Measure $N = \frac{36\ln{d}}{\delta^2}$ copies of $\ket{x}$ in the standard basis and count $n_i$, the number of times the outcome $i$ is observed. Store $\sqrt{p_i} = \sqrt{n_i/N}$ in QRAM data structure. 
\STATE Create $N = \frac{36\ln{d}}{\delta^2}$ copies of the state $\frac{1}{\sqrt{2}}\ket{0}\sum_{i \in [d]}x_i\ket{i} + \frac{1}{\sqrt{2}}\ket{1}\sum_{i \in [d]}\sqrt{p_i}\ket{i}$. 
\STATE Apply an Hadamard gate on the first qubit to obtain\\
$$
\ket{\phi} = \frac{1}{2}\sum_{i \in [d]}\left( (x_i + \sqrt{p_i})\ket{0,i} + (x_i - \sqrt{p_i})\ket{1,i}\right)
$$
 \STATE Measure both registers of each copy in the standard basis, and count $n(0,i)$ the number of time the outcome $(0,i)$ is observed. 
 \STATE Set $\sigma(i)=+1$ if $n(0,i) > 0.4Np_i$ and $\sigma(i)=-1$ otherwise. 
 \STATE Output the unit vector $\widetilde{X}$ such that $\forall i \in [N], \widetilde{X}_i = \sigma_i\sqrt{p_i}$

\end{algorithmic}
\end{algorithm}

\noindent The following version of the Chernoff Bound will be used for analysis of algorithm \ref{alg:tom}. 

\begin{theorem} (Chernoff Bound)
Let $X_j$, for $j \in [N]$, be independent random variables such that $X_j \in [0,1]$ and let $X = \sum_{j \in [N]}X_j$. We have the three following inqualities:
\begin{enumerate}
  \item For $0<\beta<1,  \PP[ X < (1-\beta)\E[X]] \leq e^{-\beta^2\E[X]/2}$
  \item For $\beta>0,  \PP[ X > (1+\beta)\E[X]] \leq e^{-\frac{\beta^2}{2+\beta}\E[X]}$
  \item For $0<\beta<1,  \PP[ |X - \E[X]| \geq \beta\E[X]] \leq e^{-\beta^2\E[X]/3}$, by composing $1.$ and $2.$\\
\end{enumerate}
\end{theorem} 

\begin{theorem} 
Algorithm \ref{alg:tom} produces an estimate $\widetilde{X} \in \R^{d}$ such that $\norm{\widetilde{X}-x}_{\infty} < (1+\sqrt{2}) \delta$ with probability at least $1 - \frac{1}{ d^{0.83}}$. 
\end{theorem} 
\begin{proof} 
\noindent Proving $\norm{x-\widetilde{X}}_{\infty} \leq O(\delta)$ is equivalent to show that for all $i \in [d]$, we have $|x_i - \widetilde{X}_i| = |x_i - \sigma(i)\sqrt{p_i}| \leq O(\delta)$. Let $S$ be the set of indices defined by $S = \{i \in [d] ; |x_i|>\delta\} $. We will separate the proof for the two cases where $i \in S$ and $i \notin S$.

\paragraph{Case 1 : $i \in S$. \\}
We will show that if  $i \in S$, we correctly have $\sigma(i) = sgn(x_i)$ with high probability. Therefore  we will need to bound $|x_i - \sigma(i)\sqrt{p_i}| = ||x_i| - \sqrt{p_i}|$.

We suppose that $x_i>0$. The value of $\sigma(i)$ correctly determines $sgn(x_i)$ if the number of times we have measured $(0,i)$ at Step 4 is more than half of the measurements, \emph{i.e.} $n(0,i) > \frac{1}{2}\E[n(0,i)]$. If $x_i<0$, the same arguments holds for $n(1,i)$. We consider the random variable that represents the outcome of a measurement on state $\ket{\phi}$. The Chernoff Bound (part 1) with $\beta=1/2$ gives

\begin{equation}\label{eq:exp}
\PP[n(0,i)\leq \frac{1}{2}\E[n(0,i)]]\leq e^{-\E[n(0,i)]/8}
\end{equation}
\noindent From the definition of $\ket{\phi}$ we have $\E[n(0,i)] = \frac{N}{4}(x_i +\sqrt{p_i})^2$. We will lower bound this value with the following argument.

For the $k^{th}$ measurement of $\ket{x}$, with $k \in [N]$, let $X_k$ be a random variable such that $X_k=1$ if the outcome is $i$, and $0$ otherwise. We define $X = \sum_{k \in [N]}X_k$. Note that $X = n_i = N p_i$ and $\E[X] = N x_i^2$. We can apply the Chernoff Bound, part 3 on $X$ for $\beta = 1/2$ to obtain, 

\begin{equation}\label{eq:chernoff3}
\PP[|X - \E[X]| \geq  \E[X]/2] \leq e^{-\E[X]/12}
\end{equation}

$$
\PP[|x_i^2 - p_i| \geq x_i^2/2] \leq e^{-Nx^2_i/12}
$$

We have $N = \frac{36 \ln{d}}{\delta^2}$ and by assumption $x_i^2 > \delta^2$ (since $i \in S$). Therefore,

$$
\PP[|x_i^2 - p_i| \geq x_i^2/2] \leq e^{-36\ln{d}/12} = 1/d^3
$$

This proves that the event $|x_i^2 - p_i| \leq x_i^2/2$ occurs with probability at least $1-\frac{1}{d^3}$ if $i \in S$. This previous inequality is equivalent to $\sqrt{2p_i/3} \leq |x_i| \leq \sqrt{2p_i}$. Thus, with high probability we have 
$\E[n(0,i)] = \frac{N}{4}(x_i +\sqrt{p_i})^2 \geq 0.82 N p_i$, since $\sqrt{2p_i/3} \leq |x_i|$. Moreover, since $|p_i| \leq x_i^2/2$, $\E[n(0,i)] \geq 0.82 Nx_i^2/2 \geq 14.7\ln{d}$. Therefore, equation \eqref{eq:exp} becomes

$$
\PP[n(0,i)\leq 0.41 N p_i]\leq e^{-1.83\ln{d}} = 1/d^{1.83}
$$
We conclude that for $i \in S$, if $n(0,i) > 0.41 N p_i$, the sign of $x_i$ is determined correctly by $\sigma(i)$ with high probability $1- \frac{1}{d^{1.83}}$, as indicated in Step 5. 

We finally show $|x_i - \sigma(i)\sqrt{p_i}| = ||x_i| - \sqrt{p_i}|$ is bounded. Again by the Chernoff Bound (3.) we have, for $0<\beta<1$:

$$
\PP[|x_i^2 - p_i| \geq \beta x_i^2] \leq e^{\beta^2Nx^2_i/3}
$$

By the identity $|x_i^2 - p_i| = (|x_i| - \sqrt{p_i})(|x_i| + \sqrt{p_i}) $ we have

$$
\PP\left[\Big||x_i| - \sqrt{p_i}\Big| \geq \beta \frac{x_i^2}{|x_i| + \sqrt{p_i}}\right] \leq e^{\beta^2Nx^2_i/3}
$$

Since $\sqrt{p_i} > 0$, we have $\beta \frac{x_i^2}{|x_i| + \sqrt{p_i}} \leq \beta \frac{x_i^2}{|x_i|} = \beta |x_i|$, therefore $\PP\left[\Big||x_i| - \sqrt{p_i}\Big| \geq \beta |x_i|\right] \leq e^{\beta^2Nx^2_i/3}$. Finally, by chosing $\beta = \delta/|x_i| < 1$ we have

$$
\PP\left[\Big||x_i| - \sqrt{p_i}\Big| \geq \delta \right] \leq e^{36\ln{d}/3} = 1/d^{12}
$$

We conclude that, if $i \in S$, we have $|x_i - \tilde{X}_i| \leq \delta$ with high probability. 

Since $|S| \leq d$, the probability for this result to be true for all $i \in S$ is $1 - \frac{1}{d^{0.83}}$. This follows from the Union Bound on the correctness of $\sigma(i)$.

\paragraph{Case 2 : $i \notin S$. \\}
If $i \notin S$, we need to separate again in two cases. When the estimated sign is wrong, \emph{i.e.} $\sigma(i) = -sgn(x_i)$, we have to bound $|x_i - \sigma(i)\sqrt{p_i}| = ||x_i| + \sqrt{p_i}|$. On the contrary, if it is correct, \emph{i.e.} $\sigma(i) = sgn(x_i)$, we have to bound $|x_i - \sigma(i)\sqrt{p_i}| = ||x_i| - \sqrt{p_i}| \leq ||x_i| + \sqrt{p_i}|$. Therefore only one bound is necessary. 

We use Chernoff Bound (2.) on the random variable X with $\beta > 0$ to obtain

$$
\PP[ p_i > (1+\beta)x_i^2] \leq e^{\frac{\beta^2}{2+\beta}Nx_i^2}
$$

We chose $\beta = \delta^2/x_i^2 $ and obtain $\PP[ p_i > x_i^2 + \delta^2] \leq e^{\frac{\delta^4}{3\delta^2}N} = 1/d^{12}$. Therefore, if $i \notin S$, with very high probability $1 - \frac{1}{d^{12}}$ we have  $p_i \leq x_i^2 + \delta^2 \leq 2\delta^2$. We can conclude and bound the error:
$$
|x_i - \tilde{X}_i| \leq ||x_i| + \sqrt{p_i}| \leq \delta+\sqrt{2}\delta = (1+\sqrt{2})\delta
$$
Since $|\overline{S}| \leq d$, the probability for this result to be true for all $i \notin S$ is $1 - \frac{1}{d^{11}}$. This follows from applying the Union Bound on the event $p_i > x_i^2 + \delta^2$. 
\end{proof}

\section{Quantum-Inspired Classical Algorithm}

Recent work \cite{tang2018quantum} has provided quantum inspired classical algorithms for clustering that rely on $\ell_2$ sampling using classical analogs of the binary search tree (BST) data structure to efficiently estimate inner products. Indeed the inner product can be efficiently approximated classically, analogous to quantum inner product estimation. As shown in \cite{kerenidis2018neural}, if $x,y \in \R^{n}$ are stored in $\ell_2$-BST then, with probability at least $1 - \Delta$, a value $s$ that approximates the inner product $\braket{x}{y}$ can be computed with the following guarantees, 
\begin{equation}
|s - \braket{x}{y}| \leq 
\begin{cases}
\epsilon \qquad\qquad \text{in time} \quad  \widetilde{O}\left( \frac{\log(1/\Delta)}{\epsilon^2}\frac{\norm{x}^2\norm{y}^2}{|\braket{x}{y}|} \right) \\
\epsilon |\braket{x}{y}| \quad \text{in time} \quad  \widetilde{O}\left( \frac{\log(1/\Delta)}{\epsilon^2}\norm{x}^2\norm{y}^2 \right)
\end{cases}
\end{equation}

The running time is similar to the quantum inner product estimation presented in Section \ref{innerproduct}, but presents a quadratic overhead on the precision $\epsilon$ and the norms of the vectors $x$ and $y$, which in our case would be $A^{\ell}_{p}$ and $F^{\ell}_{q}$, input and kernel vectors. Similarly, the steps corresponding to the amplitude amplification of Section \ref{amplitudeamplification} can be done classically with a quadratically worse dependance on the parameters. 

It is therefore possible to define a classical algorithm inspired from this work, if the matrices $A^{\ell}$ and $F^{\ell}$ are stored in classical $\ell_2$ BST. Using the above result, and applying classically non linearity and pooling, would give a forward pass algorithm with running time,  

\begin{equation}
\widetilde{O}\left(H^{\ell +1}W^{\ell +1}D^{\ell+1} \cdot \frac{M^2 C}{\epsilon^2 \mathbb{E}(f(\overline{X}^{\ell+1}))} \right).
\end{equation}

Similar to the quantum algorithm's running time (\ref{final_runtime_2}), we obtain a poly-logarithmic dependance on the size of the kernels. We however see a quadratically worse dependance with respect to $\epsilon$, $M = \max_{p,q}{\norm{A_{p}}\norm{F_{q}}}$, $C$ the upper bound of the non linearity, and the average value of $f(\overline{X}^{\ell+1})$, too. Recent numerical experiments \cite{arrazola2019quantum, kerenidis2018neural} showed that such quantum inspired algorithms were less performant than the quantum ones, and even than the standard classical algorithms for performing the same tasks. 

Finally, the quantum inspired algorithm doesn't provide the speedup characterized by $\sigma \in [0,1]$, the fraction of sampled elements among $H^{\ell +1}W^{\ell +1}D^{\ell+1}$ of them. Indeed, the probabilistic importance sampling described in Section \ref{tomographyconvolution}, that allows to sample the highest values of the convolution product output does not have a classical analog. The importance sampling does not offer an asymptotic speedup, however it could offer constant factor savings that are relevant in practice.

\end{document}